\documentclass{article}
\usepackage{fullpage}
\usepackage{hyperref}
\usepackage{verbatim}

\usepackage{colortbl}

\usepackage{paralist}
\usepackage{pdfpages}
\usepackage{enumitem}
\usepackage{booktabs}
\usepackage{color,xcolor}
\usepackage{graphicx,color,eso-pic}
\usepackage{amsmath,amssymb,stmaryrd}
\usepackage{amsmath,amsthm,amstext,amssymb,amsfonts,latexsym}

\usepackage{tikz}
\usetikzlibrary{arrows,automata}
\usetikzlibrary{shapes}
\usetikzlibrary{decorations.pathmorphing}
\usetikzlibrary{decorations.pathreplacing}

\usepackage[
	lambda,
	advantage,
	operators,
	sets,
	adversary,
	landau,
	probability,
	notions,
	logic,
	ff,
	mm,
	primitives,
	events,
	complexity,
	asymptotics,
	keys]
{cryptocode}

\newcommand{\pr}[1]{\text{Pr} \left[#1 \right]}

\renewcommand{\vec}[1]{\overrightarrow{#1}}

\usepackage{boxedminipage}
\usepackage{url}
\usepackage{graphicx,proof}
\usepackage{graphicx}
\usepackage[bf]{caption}
\usepackage{subcaption}
\usepackage{color}
\usepackage{xspace}
\usepackage{multirow}
\usepackage{amsmath,amstext,amssymb,amsfonts,latexsym}
\usepackage{wrapfig}

{
\newtheorem{theorem}{Theorem}

\newtheorem{corollary}{Corollary}
\newtheorem{definition}{Definition}

\newtheorem{claim*}{Claim}

}

\newcommand{\A}{\mathcal{A}}

\newcommand{\Bit}{\lbrace 0,1 \rbrace}

\newcommand{\Gen}{\mathsf{Gen}}

\newcommand{\Ideal}{\mathsf{Ideal}}

\newcommand{\braces}[1]{\lbrace #1 \rbrace}

\renewcommand{\sk}{\mathsf{sk}}

\newcommand{\Watermark}{\mathsf{Watermark}}
\newcommand{\Detect}{\mathsf{Detect}}
\newcommand{\PRC}{\mathsf{PRC}}
\newcommand{\Encode}{\mathsf{Encode}}
\newcommand{\Decode}{\mathsf{Decode}}
\newcommand{\prompt}{\mathsf{prompt}}
\newcommand{\tokens}{\vec{t}}
\newcommand{\End}{\mathsf{End}}
\newcommand{\Ber}{\mathsf{Ber}}
\newcommand{\Generate}{\mathsf{Generate}}
\newcommand{\Unif}{\mathsf{Unif}}

\newcommand{\Nat}{\mathbb{N}}
\renewcommand{\H}{\mathcal{H}}

\newcommand{\E}{\mathcal{E}}

\newcommand{\V}{\mathcal{V}}
\renewcommand{\O}{\mathcal{O}}

\newcommand{\M}{\mathcal{M}}

\newlength{\boxwidth}
\newcommand{\pprotocol}[5]{
\begin{figure}[#4]
\renewcommand{\arraystretch}{0}
\begin{center}
\fbox{
\begin{minipage}{0.9\textwidth}
  \textsc{#1}

  { #5}

\end{minipage}
}%fbox
\end{center}
\vspace{-0.1in}
\caption{\label{#3} #2}
\end{figure}
} 

\newcommand{\TODO}[1]{\medskip \noindent
\colorbox{red!20}{
\parbox{0.98\textwidth}{
[TODO: #1]
}}
}

\newenvironment{boxfig}[2]{% {#1}{#2} = {Caption}{label}
     \begin{figure}
     \newcommand{\FigCaption}{#1}
     \newcommand{\FigLabel}{#2}
       \vspace{-0.8cm}
     \begin{center}
       \begin{small}
         \begin{tabular}{@{}|@{~~}l@{~~}|@{}}
           \hline
           \rule[-1.5ex]{0pt}{1ex}\begin{minipage}[b]{0.95\linewidth}
             \vspace{1ex}
             }{%
           \end{minipage}\\
           \hline
         \end{tabular}
       \end{small}
       \vspace{-0.1cm}
      \caption{\FigCaption}
       \label{fig:\FigLabel}
     \end{center}
     \vspace{-0.4cm}
   \end{figure}
}

\newcommand{\Nnote}[1]{{\color{purple}[\textbf{Noam}: #1]}}
\newcommand{\Anote}[1]{}
\newcommand{\Rnote}[1]{}

\newcommand{\protocol}[4]{
\pprotocol{#1}{#2}{#3}{ht}{#4} }

\author{Noam Mazor \thanks{New York University. Research was partly done while visiting the Simons institute.}\and Andrew Morgan \thanks{Cornell Tech. Research supported by AFOSR Award FA9550-24-1-0267.}
 \and Rafael Pass \thanks{Technion, Cornell Tech and Tel-Aviv University. Research supported in part by AFOSR Award FA9550-23-1-0312,  AFOSR Award FA9550-24-1-0267,  ISF Award 2338/23 and ERC Advanced Grant KolmoCrypt - 101142322. Any opinions, findings and conclusions or recommendations expressed in this material are those of the author(s) and do not necessarily reflect the views of the United States Government, the AFOSR, the European Union or the European Research Council Executive Agency.}}

\title{Can we Watermark Low-Entropy LLM Outputs?}

\begin{document}

\maketitle

\begin{abstract}
A recent and exciting thread of work focuses on developing methods for watermarking the output of large language models (LLMs). We focus on \emph{provably undetectable} watermarking---that is, schemes that do not alter the output distribution of the LLM, yet enable embedding a watermark in the output that identifies the output as having been generated by the particular LLM. Furthermore, the watermark should be hard to remove by an adversary that may potentially edit, insert, or delete tokens from the watermarked output.

Indeed, recent work (Christ et al. [COLT '24],  Christ et al. [CRYPTO’24], Golowich et al. [NeuroIPS’24]) shows how to develop such schemes that are \emph{robust} against a constant fraction of substitutions, or even against a constant fraction of arbitrary edits.

These works, however, make strong assumptions on the amount of entropy present in the output of the LLM. Most notably, they all require \emph{constant entropy rate}---that is, a constant fraction of the tokens in a sufficiently long substring of the output need to have empirical entropy at least $O(\log |T|)$, where $T$ is the alphabet of tokens, and Golowich et al. additionally require $T$ to be larger than the security parameter.

In this work, we consider the question of whether we can also watermark the outputs of LLMs when the per-token entropy is just a constant, discarding the dependence on the alphabet size or security parameter. In this regime, we construct:
\begin{itemize}
\item A watermarking scheme robust against random substitutions (assuming subexponential LPN, as in Christ et al. [CRYPTO’24])
\item A watermarking scheme robust against random substitutions and random deletions, given either the additional heuristic assumption that the output of the LLM only introduces random errors (analogous to the assumption made by Christ et al. [CRYPTO’24]) or a construction of a pseudorandom error-correcting code robust to adversarial substitutions and random deletions.
\end{itemize}

\end{abstract}
\newpage
\section{Introduction}

With the advent of large language models (LLMs) and other tools for AI-based content generation in everyday society, and with recent advances making AI-generated content increasingly natural and human-like, retaining the ability to differentiate content generated by an AI from content generated by a human has become a prevalent and concerning issue. A natural approach to solving this problem---and indeed, an approach which multiple developers of LLMs have pledged to implement \cite{openai}, is that of \emph{watermarking}, wherein an LLM or other AI content generation algorithm is subtly altered to embed a hidden code, or ``watermark" into the content it generates, in such a way that the code can be extracted later to detect that it was in fact generated by the watermarked AI.

There are several intuitive desiderata in such a scheme. First, the watermark should be \emph{undetectable}, meaning that the watermarked content generation algorithm should not produce content that is noticeably different from content that would be generated with the same inputs but without the addition of the watermarking; that is, the distribution of the AI's outputs should not be biased by the addition of watermarking. This is usually accomplished by use of a \emph{secret key} for watermark detection, such that the detection algorithm can use it to efficiently detect the watermark but an adversary not knowing the key cannot. Additionally, a watermark should be \emph{robust} against adversaries who may wish to ``hide" that watermarked content was AI-generated by editing the output in an attempt to remove the watermark; an ideal watermark should still be detectable even in an output edited via some combination of substitutions, insertions, or deletions from an adversary.

We focus on the setting of watermarking for LLMs, where the increasing prevalence and quality of practical LLMs has unsurprisingly prompted a wide variety of ideas for watermarking schemes and techniques. Works such as \cite{cgmr24,kgwk23,kthl24,zalw23} provide strong and empirically verifiable robustness guaranties, but do so by introducing a limited amount of bias into the distribution of the LLM's output, which is compounded over repeated use.

\paragraph{Undetectable Watermarking for LLMs.}
A stronger notion of undetectability was introduced by Christ, Gunn and Zamir \cite{cgz24}; this notion requires watermarked outputs of the LLM to be \emph{computationally indistinguishable} from the original output of the LLM.
Such watermarking schemes can be used repeatedly while still remaining undetectable (and without deteriorating the performance of the LLM). In this paper, we focus on such schemes.

\cite{cgz24} provide the first scheme that achieves this notion (assuming a minimal cryptographic hardness assumption of a one-way function); their scheme, however, requires large substrings (i.e., blocks) of text generated by the LLM to remain completely unchanged for the watermark to be recoverable. In other words, robustness only holds with respect to very limited transformations (e.g., not even changes to a small \emph{fraction} of tokens).

Christ and Gunn \cite{ChrGun24} suggest a general approach for constructing watermarking schemes that satisfy both undetectability and robustness against larger classes of attack---and in particular those that can handle modification to a \emph{constant fraction} of tokens---in two steps. First, one needs to construct a so-called \emph{pseudorandom code} (PRC): a code which is both \emph{pseudorandom}, in that a codeword is indistinguishable from true randomness to a computationally bounded adversary, and \emph{robust}, in that, similar to an error-correcting code, a codeword can be recovered even after being manipulated in certain ways. Given this, a PRC codeword can be embedded in the output of the LLM to watermark it, and the watermark can then be detected by attempting to recover the codeword from the output. The pseudorandomness of the PRC provides the undetectability of the watermarking, while the watermarking scheme inherits the robustness from the code (modulo any errors which may be introduced by the embedding process). 

Christ and Gunn then construct two PRCs with a binary alphabet: one which is robust against adversarial substitutions (assuming the subexponential hardness of LPN), and one which is robust against random deletions and substitutions. They also suggest an embedding scheme, specifically for LLMs with a binary alphabet of possible tokens, that operates by attempting to correlate the bit selected by the LLM with the respective bit of the PRC codeword whenever it is possible to do so without biasing the distribution of the LLM's selection. This, depending on which of their PRC constructions are leveraged, results in a watermarking scheme which is either secure against random substitutions, or, under the assumption that the errors introduced by the embedding procedure are random, against random deletions and substitutions; additionally, it can be generalized to watermark LLMs with non-binary sets of tokens by simply using binary encoding of the tokens. 

The recent work of Golowich and Moitra \cite{gm24} constructs a PRC robust against edit distance (i.e., adversarial deletions, insertions, and substitutions). In fact, they construct such a code robust to a surprisingly large fraction ($1-o(1)$) of adversarial edits, with the caveat that it requires the size of the codewords' alphabet to be polynomially large in the security parameter.\footnote{Specifically, each token in the alphabet needs to correspond to a different index of a codeword in an underlying binary PRC.} They also present an embedding scheme for such codes, which works by assuming the LLM has a polynomially large alphabet of tokens, hashing these tokens to random tokens in the PRC's alphabet, and applying a similar approach to \cite{ChrGun24} to attempt to correlate the hashed tokens to the respective tokens in the codeword. Through this, they obtain a watermarking scheme with a stronger robustness guarantee, additionally circumventing the requirement of a heuristic assumption such as that of \cite{ChrGun24} because of the added robustness in the PRC, but at the cost of requiring a much larger alphabet of tokens (as, in particular, the PRC does not retain its robustness guarantees through reduction to binary).

Notably, however, the embedding schemes in both of these works are not guaranteed to perfectly embed the codeword; they require that the output of the LLM has high enough \emph{empirical entropy}---that is, that a sufficient number of tokens are drawn from sufficiently diverse distributions---as otherwise it will be impossible to reliably correlate the LLM output (e.g., if it is a single fixed token) with a token in a pseudorandom codeword, creating substitution errors in the watermarking process and potentially interfering with robustness guarantees.

The scheme of \cite{gm24} requires a constant fraction of the tokens in a sufficiently long substring of the LLM's output to have empirical entropy $O(\log |T|)$, where $T$ is the set of possible tokens. (This assumption on the entropy of the tokens is somewhat inherent in their approach; indeed, when the entropy of each token is 1, one could not expect to have robustness against more than $1/2$ fraction of adversarial substitutions.\footnote{For example, imagine a scenario in which we ask the LLM to output a sequence of tokens of length $2L$, each is ``Zero" with probability $1/2$ or ``One" with probability $1/2$. In this case, each token has entropy $1$ given any prefix. On the other hand, by changing $1/2$ fraction of the tokens, we can change the output sequence to be any of $\sum_{i\le L}{2L \choose i} = 2^{2L}/2$ possible sequences of  distance at most $L$ from the LLM output. 
This means that either the watermarked can be removed by changing $L$ tokens, or the scheme recognize a random sequence as watermarked with
probability at least $1/2$, contradicting soundness.}) Notably, because $|T|$ is required to be a large polynomial in the security parameter $\lambda$, this implies that the required entropy is $\Omega(\log \lambda)$. 

The entropy requirement of \cite{ChrGun24} is more moderate, but requires additional assumptions: when the LLM's output is binary, it merely requires that a constant fraction of the tokens in a sufficiently long substring have $\Omega(1)$ empirical entropy. When translating a model with a large set $T$ of potential tokens to binary, the na\"ive way to do it is to encode each token using $\log |T|$ bits. 
Using this approach, the entropy requirement of \cite{ChrGun24} would once again be $\Omega(\log |T|)$, but the model of random errors no longer makes sense if using this naive encoding (since random errors to tokens do not correspond to random errors in bit representation).
As the authors suggest, one can overcome these issues by using a Huffman encoding of the token, according to some \emph{a priori} distribution over $T$, and simply using the first bit of the encoding to embed the binary PRC.
This, however, introduces the additional strong assumption 
that the encoder knows the a priori distribution of the LLM's output on $T$ and that the output of the LLM is distributed closely to it (so that 
the first bit of the Huffman encoding of the token has high empirical entropy).

Thus, unless we make strong distributional assumptions on the output of the LLM, neither of the above approaches can be used when the entropy of most tokens is a small constant, and in particular when the entropy is close to a single bit (or even smaller).\footnote{Note that the approach of \cite{ChrGun24} works when the entropy of each \emph{binary} token is one. %\Rnote{remove: or when the a priori distribution of the tokens has low constant entropy}. 
We here want to consider the realistic case where the token set is significantly larger, yet each token has low entropy within that larger alphabet.} Thus, in this work, we address the following question: 

\begin{quote}
\emph{Is there an undetectable and robust watermarking scheme for LLM outputs with a large alphabet of tokens that preserves robustness even when almost all tokens have low entropy (e.g., entropy 1)?}
\end{quote}

We remark that when only requiring an undetectable watermarking scheme for LLM with very limited robustness guarantees (i.e., the requirement that a sufficiently large enough block of the output is completely unchanged),
then the original scheme of \cite{cgz24} indeed supports this low-entropy regime---in fact, they handle even sub-constant entropy. However, as far as we know, no known schemes provide any non-trivial notion of robustness with constant or better entropy.

% we already have one nutshell, just calling this "our results" for now but will maybe change if i think of something better

\paragraph{Highlights of our results.} Our main results will be to provide initial constructions of such watermarking schemes. Based on essentially the same hardness assumptions as those in \cite{ChrGun24}, we construct a watermarking scheme which provides robustness guarantees against a broad class of random substitution attacks (and, given additional assumptions, random deletions) for even outputs with very low entropy; specifically, our scheme's only requirement for robustness is that a constant fraction of tokens in a sufficiently long substring of the output have a single bit of empirical entropy. An additional advantage of our approach is that, because our watermarking scheme models token generation at the token level rather than reducing it to a binary alphabet as \cite{ChrGun24} does, our robustness guarantees, similarly to those of \cite{cgz24} in the high-entropy regime, deal with a more practical and realistic class of attacks that modify the output at a token level (e.g., paraphrasing by replacing tokens with synonyms, or redacting tokens by replacing them with something deterministic); in contrast, the practical significance of attacks against a binary translation of the output tokens is far less clear.
%, especially in the case of a variable-length encoding (e.g., Huffman) where a single change in the binary translation could change the entire untranslated output or vice versa. We discuss this in further detail in section \ref{sec:comparison}.

%\TODO{
%Other works: yevgeniy, CCA security.

%Explain smiley attack: it’s ok, we do not consider adversarial %queries (they degrade the performance of the output)}

\subsection{Results}

%\TODO{yevgeniy/miranda paper deals with randomized ``emoji attack", but so do we (it falls within our model of random substitutions).}

We present, for LLMs with a large token set, a PRC-based watermarking scheme that provides similar robustness guarantees to the scheme of \cite{ChrGun24} with a much weaker entropy requirement. In particular, in contrast to prior works which require output tokens (or substrings of tokens) to have an amount of average empirical entropy that is at least logarithmic in the size of the alphabet of tokens, we show that even outputs satisfying only a very weak notion of ``minimal empirical entropy"---which requires only that a constant fraction of the tokens in a sufficiently long substring of the output have a single bit of empirical entropy, independently of the size of the token alphabet---can be robustly watermarked without additional distributional assumptions.

We first consider the setting of robustness against random substitutions. Notably, it is not immediately clear how a ``random substitution" is even defined in the case of an output that uses a large token alphabet rather than restricting to a binary setting where a substitution always implies a bit flip; we model such a transformation by allowing a theoretical adversary to a priori map each token in the alphabet to a distribution of ``replacement" tokens, and subsequently replacing randomly selected tokens with a sample from the respective replacement distribution.\footnote{This definition, as mentioned, captures several natural notions of substitution-based attacks such as paraphrasing or redaction.} We show the following:
\begin{theorem}[Informal]
If there exists a binary PRC which provides robustness against a $(1/2 - \delta)$ fraction of adversarial substitutions for any constant $\delta > 0$, then there exists a large token alphabet watermarking scheme that is computationally undetectable and, for any LLM output satisfying the minimal empirical entropy requirement, robust to a $(1/2 - \epsilon)$ fraction of random substitutions for any constant $\epsilon > 0$.
\end{theorem}

Combining this with the PRC construction of \cite{ChrGun24}, we obtain:

\begin{corollary}[Informal]
Assuming subexponential hardness of LPN, there exists a large token alphabet watermarking scheme that is computationally undetectable and, for any LLM output satisfying the minimal empirical entropy requirement, robust to a $(1/2 - \epsilon)$ fraction of random substitutions for any constant $\epsilon > 0$.
\end{corollary}

Additionally, we can extend this result to provide robustness to random substitutions and deletions, but, as in \cite{ChrGun24}, a PRC that is robust to only \emph{random} substitutions and random deletions is insufficient on its own. 
\cite{ChrGun24} leverages such a PRC but also relies on an additional heuristic assumption---to which the authors refer as the output text being ``sufficiently variable"---which requires that the distribution of the errors to the PRC codeword introduced by the watermarking algorithm is ``consistent" in that the errors can be modeled as a binary symmetric channel. We generalize this assumption to the case of a large token alphabet (which, notably, requires a generalization of the definition of a random substitution compared to the binary-alphabet case) by introducing the notion of a ``natural" prompt, or a prompt given to the LLM that, over the distribution of all possible outputs and subsequent random substitution attacks, causes the errors to the PRC codeword introduced by the composition of the watermarking algorithm and the random substitution channel to likewise be modelable by a binary symmetric channel. Given this, we show:

\iffalse
\cite{ChrGun24} circumvents this issue using a heuristic assumption that the distribution of the errors to the PRC codeword introduced by the watermarking algorithm is ``consistent" in that the errors can be modeled as a binary symmetric channel. In our case, however, since using a larger token alphabet introduces intricacies with the introduction of random substitution errors that are not present with a binary alphabet, we require a heuristic assumption that generalizes the one they use to construct deletion-robust watermarking in their work, which amounts to the errors to the PRC codeword introduced by the composition of the watermarking algorithm and the random substitution channel being ``consistent" in the same way. We show the following:
\fi

%\TODO{quantify assumption over prompts: ``for every prompt where the output distribution produces consistently distributed errors" / ``assuming nice prompts"}

\begin{theorem}[Informal]
If there exists a binary PRC which provides robustness against a $(1/2 - \delta)$ fraction of adversarial substitutions and a $(1 - \delta')$ fraction of random deletions for any constants $\delta, \delta' > 0$, then there exists a large token alphabet watermarking scheme that is computationally undetectable and, for any natural prompt and LLM output satisfying the minimal empirical entropy requirement, robust to a $(1/2 - \epsilon)$ fraction of random substitutions and a $(1 - \epsilon')$ fraction of random deletions for any constants $\epsilon, \epsilon' > 0$.
\end{theorem}

We can again combine this with the corresponding PRC construction of \cite{ChrGun24} to obtain:

\begin{corollary}[Informal]
Assuming subexponential hardness of LPN, there exists a large token alphabet watermarking scheme that is computationally undetectable and, for any natural prompt and LLM output satisfying the minimal empirical entropy requirement, robust to a $(1/2 - \epsilon)$ fraction of random substitutions and a $(1 - \epsilon')$ fraction of random deletions for any constants $\epsilon, \epsilon' > 0$.
\end{corollary}

Importantly, this notion of a ``natural" prompt rules out certain contrived or pathological examples of prompts whose outputs are known to be inherently difficult to robustly watermark, such as the well-known ``emoji attack" \cite{kgwk23} wherein an LLM is prompted to output a response to a certain sub-prompt but insert a certain (deterministic) emoji after every token of its response. Effectively, the emoji attack performs a (structured) deletion attack on an output engineered to have a large volume of entropy-less tokens; indeed, because the emoji tokens in the pre-deletion LLM output have zero entropy, they will inherently have a higher embedding error rate than tokens with greater entropy, causing the distribution of errors to be skewed away from a ``natural" distribution.\footnote{As noted by an insightful reviewer, schemes such as ours can be easily made robust against the classical emoji attack \cite{ChrGun24}; however, variants, such as inserting emoji randomly rather than predictably after each token, or inserting an overwhelming fraction of emoji rather than 1/2, add complications against which it is not necessarily clear how to maintain robustness.}

Such attacks as the ``emoji attack" are, while theoretically effective, perhaps undesirable in practice as they simply act on a pre-existing LLM output by, while not directly degrading the quality of the response, reducing its overall quality \emph{relative to the output size}. Thus, one might hope that, in contrast, prompts that utilize the full ``intelligence" of the LLM would be close to ``natural". While this may be the case, requiring the notion of a ``perfectly natural" prompt (or even something statistically close to ``perfectly natural"), as would be necessary for our result or that of \cite{ChrGun24}, is overly restrictive to the point of being rather unrealistic, as, at a high level, producing a ``consistent error distribution" from PRC-based embedding requires tokens in the output to be uncorrelated and have consistently high entropy, whereas, in realistic outputs---for instance, the vast majority of natural language---tokens tend toward being precisely the opposite.

One way to bypass the necessity for these heuristic assumptions would be to strengthen the robustness of the underlying PRC to handle \emph{adversarial} substitutions (and random deletions) as opposed to only random substitutions:

\begin{theorem}[Informal]
If there exists a binary PRC which provides robustness against a $(1/2 - \delta)$ fraction of adversarial substitutions and a $(1 - \delta')$ fraction of random deletions for any constants $\delta, \delta' > 0$, then there exists a large token alphabet watermarking scheme that is computationally undetectable and, for any LLM output satisfying the minimal empirical entropy requirement, robust to a $(1/2 - \epsilon)$ fraction of random substitutions and a $(1 - \epsilon')$ fraction of random deletions for any constants $\epsilon, \epsilon' > 0$.
\end{theorem}

While \cite{gm24} demonstrates a PRC with the requisite robustness, their construction requires a large alphabet of tokens, making it unsuitable for our results because it would vastly increase the entropy requirement. Indeed, no known construction of binary-alphabet PRCs provides robustness against a large fraction of both adversarial substitutions and random deletions. 

However, the recent PRC construction of \cite{cggm25} allows instead for robustness against a large fraction of adversarial substitutions and a constant fraction of \emph{adversarial edits}---that is, insertions \emph{and} deletions---but with the caveat that the constant itself is very small and depends on the output's entropy\footnote{That is, it is constant only assuming that the output has a constant entropy average for a sufficiently long substring.}, and additionally that the construction itself is based on a non-standard ``permuted codes" assumption. The authors demonstrate that they can use the framework of \cite{ChrGun24} to construct binary-alphabet watermarking robust to a constant fraction of adversarial edits based on similar entropy requirements to \cite{ChrGun24}; we in turn show as our final result that we can leverage our framework with their PRC construction to construct minimal-entropy watermarking for a large token alphabet that likewise guarantees robustness to a constant fraction of adversarial edits:

\begin{theorem}[Informal]
If, for any $\delta > 0$, there exists a constant $\epsilon > 0$ and a binary PRC which provides robustness against a $(1/2 - \delta)$ fraction of adversarial substitutions and an $\epsilon$ fraction of adversarial edits, then, for any $\delta' > 0$, there exists a constant $\epsilon'$ and a large token alphabet watermarking scheme that is computationally undetectable and, for any LLM output having at least a $\delta'$ fraction of tokens with one bit of entropy, robust to an $\epsilon'$ fraction of adversarial edits.
\end{theorem}

This comes with the caveat that the size of the substring required for ``minimal entropy" is larger; in particular, since a theoretical adversary able to make a constant fraction of arbitrary edits could decide to target a sufficient-entropy substring of any output, we require that at least a (larger) constant fraction of the entire output have sufficiently many tokens with a bit of empirical entropy. Combined with the construction of \cite{cggm25}, which in particular is parameterized such that $\epsilon$ is $\Tilde{O}(\delta^3)$ in the above statement, this gives:

\begin{corollary}[Informal]
Assuming the ``permuted codes" assumption (Conjecture 3.1) of \cite{cggm25}, then, for any $\delta' > 0$, there exists a constant $\epsilon' = \Tilde{O}((\delta')^4)$ and a large token alphabet watermarking scheme that is computationally undetectable and, for any LLM output having at least a $\delta'$ fraction of tokens with one bit of entropy, robust to an $\epsilon'$ fraction of adversarial edits.
\end{corollary}

For clarity, we concisely summarize our results and provide comparison to related results in table form in Figure \ref{fig:results}.

\begin{boxfig}{Summary of our results compared to related results in terms of per-token entropy requirement (over a sufficiently long substring unless otherwise stated); robustness to (random or adversarial) substitutions, insertions, and deletions; and required assumptions.}{results}
    \begin{tabular}{|l||c||c|c|c||l|}
    \hline
     & & \multicolumn{3}{c||}{\textbf{Robustness Guarantees}} & \\
    \hline
    \textbf{Result} & \textbf{Entropy} & \textbf{Subst.} & \textbf{Edits} & \textbf{Adv.?} & \textbf{Assumptions} \\
    \hline
    \hline
    \cite{ChrGun24}    & $O(\log |T|)$      & $1/2 - \epsilon$ &  &  & subexponential LPN    \\
                         & $O(\log |T|)$      & $1/2 - \epsilon$ & $1 - \epsilon$ (del.) &  & subexponential LPN + heuristic\\
    \hline
    \cite{gm24}    & $\text{poly}(\log(\lambda))$      &  & $1 - \epsilon$ & \checkmark & local weak PRFs    \\
    \hline
    \cite{cggm25}    & $O(\log |T|)$      &  & $\epsilon$ & \checkmark & permuted codes    \\
    \hline
    \hline
    Cor. 1    & $O(1)$      & $1/2 - \epsilon$ &  &  & subexponential LPN    \\
    \hline
    Cor. 2    & $O(1)$      & $1/2 - \epsilon$ & $1 - \epsilon$ (del.) &  & subexponential LPN + heuristic\\
    \hline
    Cor. 3    & $O(1)$      & $1/2 - \epsilon$ & $1 - \epsilon$ (del.) &  & PRC with adversarial sub.+del. robustness   \\
    \hline
    Cor. 4    & $O(1)$ overall  &  & $\epsilon$ & \checkmark & permuted codes    \\
    \hline
    
\end{tabular}
\\

\end{boxfig}

\paragraph{Our approach in a nutshell.} At a high level, through using a binary PRC but allowing the generative model output to have a large token alphabet, our approach provides a natural---and advantageous---``middle ground" between the watermarking scheme of \cite{ChrGun24}, which deals with entirely binary tokens, and the watermarking scheme of \cite{gm24}, which deals entirely with a large token alphabet.

The watermarking algorithms constructed by \cite{ChrGun24} are designed for generative models which produce a binary output, and they attempt to embed the codeword of the (binary) PRC directly into the output of the model. The most direct way to transform this into watermarking for a model which outputs tokens from a larger alphabet is, as the authors propose, by simply converting the tokens into a binary representation. As we have mentioned, this comes at the cost of increasing the entropy requirement per output token proportionally to the token's binary length. However, a more concerning issue with this approach is that this transformation does not preserve the token structure of the model's output and thus dramatically changes the class of attacks against which the watermarking is robust; random substitutions or deletions of single bits in a binary encoding are incomparable with the more practical class of attacks (e.g., paraphrasing or redaction) which manipulate the output on the token level in an attempt to transform a model's watermarked output into a similar but non-watermarked output.

Our approach preserves the token structure of the output by, rather than attempting to embed the PRC codeword directly into the output, hashing the output tokens into bits and attempting to embed the codeword into the \emph{hashes} of the output tokens; as a result, the output's structure is preserved, as each token corresponds to exactly one bit of the PRC codeword. \cite{gm24} uses a similar approach, but rather than hashing the output tokens to binary and embedding a binary PRC codeword, they embed a codeword of a PRC with a polynomial-sized alphabet and hash the output tokens into that alphabet. This allows them to obtain significantly stronger robustness guarantees, but at the cost of dramatically increasing the per-token entropy requirement to be logarithmic in the security parameter, since successful embedding requires the hash of the output token distribution to be close to uniform over the PRC's alphabet. In contrast, our approach has more limited robustness, but achieves it with a far more reasonable requirement of constant entropy because we restrict the PRC to a binary alphabet.

\subsection{Technical Overview}

We construct our watermarking scheme based on a binary PRC such as those of \cite{ChrGun24}; while the large alphabet PRC of \cite{gm24} features very strong robustness guarantees, the requirement of a large alphabet makes it ill-suited to the constant-entropy case we wish to deal with. The high level idea is to reduce a large alphabet output into binary by hashing each token using a random hash function $h\colon T \to \Bit$, and then embedding the binary PRC codeword in the binary hashed output, using the same scheme as \cite{ChrGun24}. The main difficulty is to choose the hash function in a way that (1) preserves the entropy, while (2) ensuring that we still can handle deletions. Using a new hash function for each token would deal with (1), but would prevent (2). We show that using the same hash function for a sufficiently long block of tokens---specifically, for each codeword of the PRC we wish to embed---will achieve a compromise that simultaneously enables dealing with (1) and (2). 

We proceed to provide a more detailed exposition of the scheme and the analysis idea.
We start by explaining how to embed a bit from the PRC into the output of the LLM, and next discuss how to appropriately choose the hash functions.

\iffalse In more detail, given a binary PRC codeword, we watermark an LLM output by embedding that codeword is to use a hash function to hash each token to a bit, after which we use a similar approach to \cite{ChrGun24} or \cite{gm24} to attempt to correlate the hash of the token with the PRC codeword with the maximum probability with which it is possible to do so without biasing the output distribution. Detection, then, can be performed by hashing the output tokens; if enough of the correct bits of the codeword were successfully embedded to satisfy the robustness of the PRC, then the codeword can be successfully detected, indicating the presence of the watermark.
\fi

\paragraph{Embedding PRC codewords in Binary Tokens (warm-up).} To explain our approach, let us first review the approach of \cite{ChrGun24} for embedding a PRC codeword in a \emph{binary} LLM. 
Let $D$ be the next token distribution of the LLM; since the LLM is binary, $D$ is  a Bernoulli distribution; such a distribution can be viewed as with some probability $p$ (``success") sampling a uniform bit, and with probability $1-p$ (``failure") picking some fixed deterministic value (0 or 1).\footnote{For example, if $D$ is $1$ with probability $\epsilon<1/2$ and $0$ otherwise, $D$ can be written as a combination of a uniform bit with probability $2\epsilon$, with the constant $0$ with probability $1-2\epsilon$. }
 Thus, since the PRC codeword we are given is pseudorandom, it suffices, in order to preserve the distribution of the LLM's output, to select a token according to the  PRC whenever the uniform distribution over $\Bit$ is sampled (i.e., the ``success" case which happens with probability $p$). If instead the fixed value is chosen by $D$ (i.e., the ``failure" case, which happens with probability $1-p$), then we must select that fixed element. This is the case in which the watermarking scheme may create an error since that element, with probability 1/2, it will not be equal to the corresponding bit of the PRC codeword.
When the average value of $p$ is sufficiently large, this only introduces a small number of errors in expectation, meaning that the watermark will be detectable by recovering the PRC codeword with overwhelming probability, even when additional errors may be introduced.

\paragraph{Dealing with Non-Binary Tokens.} 
\cite{ChrGun24} also provides a method for reducing non-binary tokens to binary: simply take the binary representation of the non-binary token. This method, however, increases the entropy requirement; while the authors suggest using a more efficient encoding such as Huffman to mitigate this, doing so would, as previously mentioned, require additional distributional assumptions that we want to avoid.\footnote{As \cite{ChrGun24} point out, even when the binary representation of the token is non-trivial this translation is not preserved under random errors; they thus suggest simply taking the first bit of the binary representation of the token with the hope that this bit has high entropy. For the Huffman encoding, this holds when each token is distributed according to the assumed a priori distribution. But when the token distribution is significantly different, this does not necessarily hold.}

Thus, we instead opt to hash non-binary tokens into a binary token, use the above approach to encode the PRC into the output of the hash $z$, and finally obtain the non-binary codeword by sampling tokens from the distribution conditioned on the hash value being $z$; the watermark can then be detected by hashing the output tokens into the encoded PRC codeword. For our purposes, and to enable our analysis, we will rely on a variant of this approach. To sample the next token from the distribution $D$, we will first sample two independent samples $t_0,t_1$ from $D$. Next, if $h(t_0) \neq h(t_1)$ (``success"), we output token $t_b$ where $h(t_b)$ is the bit of the PRC we are trying to encode. Otherwise, $h(t_0) = h(t_1)$, and we simply output a randomly chosen $t_b$ (``failure"), noting that this $t_b$ has a probability of 1/2 to incorrectly encode the corresponding PRC bit.
As we shall argue, the probability of the failure event cannot be too large, and thus this encoding will not introduce too many errors in the PRC.

\paragraph{Selecting hash functions.} Another important consideration is how we generate random hash functions $h$ for hashing tokens. A na\"ive way is to generate an independently random function $h$ to hash each token of the output, and include the description of every function $h$ used as part of the secret key used to detect the watermark. While this is simpler to analyze and works well for the case of robustness against substitutions, it is easy to see that this falls apart as soon as insertions or deletions are considered, since there is no longer a way to link an output token to the respective hash function. Other works such as \cite{kgwk23} consider alternative approaches such as using the previous token or a consecutive string of previous tokens as randomness, but these are similarly vulnerable to insertions, deletions, or other attacks.

We opt to split the output of the LLM into \emph{blocks}, each block corresponding with a different randomly generated PRC codeword\footnote{Splitting the output in this way is a common idea in watermarking schemes since the redundancy allows a watermark to be detected as long as a single PRC codeword can be recovered.}, and generate a new random hash function to be applied to \emph{every} token in a single block---in other words, we \emph{reuse} the \emph{same} hash function for every token in the block. With this reuse, we are able to deal with deletions by simply using the respective hash function to attempt to decode each specific block of the output. The problem, of course, is that it creates additional complexity in the analysis of robustness. 

Namely, if the hash function were generated independently for each token, it is easy to assert that the probability of each token of the watermarking algorithm generating an error is likewise independent from other tokens. But, with the necessary reuse of hash functions, that is no longer the case---indeed, a pathological model can attempt to create correlated errors by reusing certain tokens or distributions of tokens (e.g., distributions with a significant probability of hashing to a single value and giving $p=0$), or an adversary against robustness can abuse the fact that every instance of a commonly used token will hash to the same value by replacing every instance of that token with a different token to attempt to influence the recoverability of the PRC codewords with significant probability.

It is not clear how to circumvent the potential for this kind of adversarial behavior while still preserving robustness against deletions in the setting of a binary PRC, which is why we (and \cite{ChrGun24} before us) require either leveraging a PRC that is robust against specifically \emph{adversarial} substitutions or making an additional assumption on the model requiring the errors created to be ``random enough" to circumvent these pathological examples---the authors of \cite{ChrGun24} refer to this assumption as the output being ``sufficiently variable" or ``having few repeated words". (The strong robustness property of the large alphabet PRC of \cite{gm24} circumvents this issue by being robust to adversarial errors, but leveraging a PRC with such a large alphabet innately introduces a very strong entropy requirement.)

\paragraph{Key Insights for the Analysis.}
We proceed to provide the key ideas required for the analysis of the above scheme. Roughly speaking, our goal is to show that the above embedding procedure will introduce a small number of errors into the PRC codeword. The above assumption on the LLM, together with the fact that the remainder of the errors introduced are random, will provide that, with high probability, at least one of the PRC codewords embedded in the output will still be recoverable by hashing the (edited) output text.
Towards this, we proceed in two steps:

\paragraph{Analyzing high-entropy tokens.} First, notice that if a token $t_i$ is drawn from the output distribution $D$ and has empirical entropy at least 1, then $D(t_i) \leq 1/2$, which in particular means that the probability that a pair $(t_0, t_1)$ was drawn with $t_0 = t_1=t_i$ is likewise at most 1/2. If $t_0 = t_1$, then of course $h(t_0) = h(t_1)$, and this is the aforementioned ``failure" case where there is a 1/2 probability for the respective bit of the PRC to not be successfully encoded. On the other hand, if $t_0 \neq t_1$, then, over a randomly selected hash function $h$, there is a 1/2 probability for $h(t_0) = h(t_1)$, which once again carries a 1/2 probability over the PRC codeword for the encoding to be unsuccessful; however, if $h(t_0) \neq h(t_1)$, then we arrive in the aforementioned ``success" case where the encoding will always succeed. 

Putting this all together, we see that, given a token $t_i$ with empirical entropy of at least 1, there is at most a 3/4 probability (over randomly selected $h$) of arriving in the ``failure" case and at least a 1/4 probability of the ``success" case, meaning that the encoding has at least a 5/8 probability of success. Indeed, in the simplified case where a different hash function $h$ is drawn for every token and these probabilities are thus independent, this is sufficient to conclude (by a concentration bound such as Hoeffding's inequality) that close to 5/8 of the tokens drawn with empirical entropy at least 1 correspond to successfully embedded bits of the PRC codeword; thus, if there exist even a constant fraction $\delta$ of such tokens, the error rate of the embedding process is approximately $(1/2 - \delta / 8)$.

\paragraph{Random substitutions with a large token alphabet.} At this point, if one can show that the expected error rate remains a constant less than $1/2$ after even a $1/2 - \epsilon$ rate of random substitutions, it is straightforward to show that, leveraging a PRC with robustness against a $1/2 - O(1)$ fraction of (adversarial) substitutions, some block would be recoverable with overwhelming probability. In the model of \cite{ChrGun24} with a binary alphabet where random substitutions simply reverse each bit with some independent constant probability $p < 1/2$, this is very easy to show, as each error bit simply becomes a non-error (and vice versa) with probability $p$.

However, in our model, since the bit we embed is the hash of a token rather than the token itself, there is no guarantee that substituting a token for a different one will change its hash value, and thus we require slightly different analysis. At a high level, we leverage the fact that, because the PRC codeword is pseudorandom and each bit is 1 or 0 with probability $1/2$, any token that was drawn independently of the respective PRC bit must be an error with probability $1/2$ as well. In fact, the only time when a token is \emph{not} drawn in such a way is in the ``success" case above where $h(t_0) \neq h(t_1)$, which, if it occurs with some probability $c \geq \delta/4$, bounds the expected error rate by $(1/2 - c / 2)$ according to the argument above.

The key insight to our approach is then the observation that, as long as a ``random substitution" replaces (with a given independently drawn probability) a token with some (potentially arbitrary) distribution over tokens that is determined only by the replaced token itself and no additional context, randomly substituting a token drawn independently of the PRC codeword will always result in a token that is likewise independent of the PRC codeword, and thus must once again have exactly a $1/2$ probability to be an error---in other words, the expected number of errors generated from tokens outside of the ``success" case is always zero! And so, even in the worst possible case where every substitution on a token in the ``success" case changes it from a non-error to an error, the expected number of errors generated is still only $(1/2 - \epsilon)c < c/2$, meaning that the expected error rate remains under 1/2 as desired.

\paragraph{Dealing with reused hash functions.} In fact, the analysis above works almost directly if we select a new function $h$ for every token; however, as mentioned, while this works for the case of robustness strictly against substitutions, it fails to provide robustness against even a single insertion or deletion; thus, we turn to the alternative approach of using the same $h$ to embed an entire PRC codeword. Compared to using the same hash function for the entire output (as in \cite{gm24}), this carries the advantage that different blocks have independent probabilities of the PRC codeword being sufficiently embedded (which provides redundancy, as we only need a single block to be recoverable to detect the watermark); however, we can no longer rely on independence between the embedding probabilities of different tokens within the same block to bound the number of tokens that are able to be successfully embedded.\footnote{For example, if the model were to generate every token within a block from a uniform distribution over the same two elements, there would only be two possible outcomes, each with probability 1/2: either $h$ ``splits" the two tokens and the entire PRC codeword would be successfully embedded, or that $h$ would map both tokens to the same bit and the hash of the output would consist of only that bit.} 

Instead, we need to rely on an averaging argument to assert that, with probability 1/4 over the choice of $h$, at least 1/4 (the expected fraction) of the tokens with sufficient entropy result in $h$ being ``split" and thus the respective bit of the PRC codeword being successfully embedded; in this case, the tokens where $h$ is not ``split" still have an independent 1/2 probability for the embedding to succeed (since irrespective of what the selected token hashes to, the PRC codeword bit is pseudorandom), and so, as long as the respective block has at least a constant fraction of tokens with sufficient entropy, we again can assert that the embedding process has at most $(1/2 - \epsilon)$ error rate for that block. And since the 1/4 probability of this occurring is indeed independent between blocks (since we redraw $h$ for each block), there will with overwhelming probability still be a constant fraction of blocks with low enough error rate from embedding that the embedded code will still be recoverable even after random substitutions and deletions in the output, even if only a polynomial number of blocks (e.g., a substring containing $n$ consecutive blocks) have sufficiently many high-entropy tokens.

\paragraph{Adding deletion and edit robustness.} Given a PRC-based watermarking scheme robust against random substitutions, one could intuitively add robustness to even as much as a $1 - O(1)$ fraction of random deletions simply by leveraging an underlying PRC which is additionally robust against a similar $1 - O(1)$ rate of random deletions; the case of random substitutions is the only case where the errors introduced by the embedding process need to be added since, naturally, the errors introduced by the embedding process are strictly limited to substitutions. However, recalling that, for the case of random substitutions, we needed a PRC that was robust against specifically \emph{adversarial} substitutions (since, indeed, the errors introduced by embedding are not necessarily independently random), this would mean that robust watermarking against random substitutions and deletions would require the PRC to be robust against adversarial substitutions and random deletions; while \cite{ChrGun24} gives a construction of a PRC robust against \emph{random} substitutions and random deletions, there is no known construction of a binary \footnote{The large alphabet PRC of \cite{gm24} does provide the necessary robustness properties, but, as we have mentioned, leveraging a PRC with a large alphabet would dramatically increase the entropy requirement for the output.} PRC with this stronger notion of robustness against a large fraction of random substitutions and deletions.

\cite{ChrGun24}, in their construction of watermarking robust against random substitutions and deletions, circumvents this issue by using their PRC robust against \emph{random} substitutions and random deletions, but adding an additional assumption that the errors introduced by the embedding process are modelable as a binary symmetric channel (i.e., an independent and constant probability $p$ for each token to be incorrectly embedded into the watermarked output), an assumption that they refer to as the output tokens being ``sufficiently variable". This way, since \cite{ChrGun24} considers binary tokens and random substitutions made to the watermarked output are thus also given by a binary symmetric channel, the resulting output will be subject to the robustness guarantees of the underlying PRC since the composition of two binary symmetric channels is likewise a binary symmetric channel. 

Trying to apply this approach to our model with a large alphabet of tokens, however, immediately fails because random substitutions are defined differently; it is unclear, and additionally dependent on the hash function $h$, whether replacing a given token by a different token will even change the respective hash bit. In order to prove a similar result in our regime, we would require the assumption that, not only does the embedding process create errors in the pattern of a binary symmetric channel, but also the random substitutions made to the output tokens introduce hash bit flips in the pattern of a binary symmetric channel (with respect to the given hash function $h$ for any specific block)---alternatively, the weaker assumption that the same is true of the composition of the embedding process and the random substitutions also suffices.

The primary issue with these heuristic assumptions is that even the weaker assumption of \cite{ChrGun24}---that the embedding errors resemble a binary symmetric channel---seems unrealistic in practice. Recalling that the success rate of embedding is highly dependent on the empirical entropy of the respective token, a straightforward counterexample for this is the ``emoji attack" where a zero-entropy token (some deterministic emoji) is introduced between every pair of ``real" tokens in the output, leading to those tokens having significantly higher embedding error rates than the tokens with non-zero entropy. However, even setting aside such pathological examples, there are ubiquitous patterns in natural language (e.g., common phrases or $n$-grams, or common prefixes in LLM responses given a certain prompt) where the entropy of tokens is significantly lower in certain blocks of tokens than in others, not to mention correlated between different tokens, making it dubious that one could find sufficiently ``variable" outputs for the probability of an embedding error to be constant and independent across all tokens in even a significant portion of an output.

Thus, we opt to prove our result in two ways; the first is by using the heuristic assumption discussed above, whereas the second circumvents the need for such an assumption by instead leveraging a PRC with robustness against \emph{adversarial} substitutions and random deletions, albeit with the caveat that there is no currently known construction of such a PRC robust to a large fraction of random deletions. Once we apply either of these approaches, demonstrating robustness against random deletions in our existing watermarking algorithm becomes relatively straightforward as discussed above, since random deletions to the watermarked output translate directly to random deletions in the PRC codeword and are thus immediately subject to the corresponding definition of robustness for the PRC.

Interestingly, the recent PRC construction of \cite{cggm25} can, as the authors show, be parameterized to be robust against a $1/2 - \epsilon$ fraction of adversarial substitutions and a small ($\Tilde{O}(\epsilon^3)$) fraction of \emph{adversarial edits} (insertions or deletions), which provides a solution to the above issue without requiring additional heuristic assumptions. While we leverage this in our final result to construct similarly edit-robust watermarking, we note that this does not supersede our discussion of or results on random deletions due to the bound on the robustness parameters being quite small (and additionally the requirement on the minimal-entropy substring length increasing to a constant fraction of the output, as opposed to merely $n$ blocks of the PRC for our other results). Obtaining low-entropy large-alphabet watermarking with robustness against a significant, let alone $1-\epsilon$, constant fraction of even random deletions without the need for such heuristic assumptions remains an open question.

\subsection{Comparison with Related Work}
\label{sec:comparison}

Our work builds upon the previous undetectable watermarking schemes of \cite{ChrGun24}, \cite{cgz24}, and \cite{gm24} by providing a scheme with both a low entropy requirement for the tokens generated by the model and strong robustness guarantees against attacks that, notably, manipulate tokens in the natural output directly rather than relying on a reduction to a binary alphabet. To explain what we mean by this and why it is advantageous, recall that, in the results of \cite{ChrGun24}, which focus on watermarking over a binary alphabet (i.e., every token is 0 or 1), one can simply model a $p$ fraction of random substitutions as a binary symmetric channel, i.e., flipping each bit independently with probability $p$. 

To extend their results to an output drawing from a large alphabet of tokens, the authors propose translating each token to a binary encoding, either by using a fixed-length encoding or a more efficient variable-length (e.g., Huffman) encoding, and generating a watermarked binary (translated) output by using a modified generative model that generates the binary translations of its output tokens. While they indeed prove robustness against random \emph{binary} substitutions to bits in the translated output, this definition of robustness loses its meaning entirely when one considers the more natural definition of substitutions made to tokens in the untranslated output, as mentioned in \cite{ChrGun24}. For instance, even using a fixed-length encoding, random binary substitutions change the corresponding untranslated tokens in ways that are not necessarily random and depend on which bits are affected; worse still, using a variable-length encoding means that a single random substitution to the translated binary output can potentially alter \emph{every} untranslated token in the entire output. In the other direction, since each token in the untranslated output corresponds to multiple bits in the binary translation, changing any token will correspond to modifying a correlated block of bits in a way that is not necessarily even independent between the bits in the block.

In contrast to \cite{ChrGun24}, we show robustness against modifications specifically to tokens in the output, which we remark is a far more natural definition as it preserves the structural features of the output that would be discarded by manipulating a translation of the output to a binary encoding. While it is not initially clear how one would even define a ``random substitution" for the case of a large set of tokens---in particular, while it is fairly straightforward to start with a similar definition to a binary symmetric channel by modifying each token with some independent probability $p$, it is less clear how each token would be modified given a large alphabet of tokens---we show that we can achieve robustness against even a class of random substitutions where an adversary may, a priori (i.e., before seeing the output), choose, for each token in the alphabet, an arbitrary distribution of tokens to replace an occurrence of that token if it is selected randomly to be substituted. This captures some basic semantic attacks such as limited forms of substitution or paraphrasing where an adversary may opt to replace some randomly chosen words with synonyms in an attempt to remove a watermark, giving robustness to a variety of practical attacks.% whereas the class of attacks on a binary translation of the model's output (as considered by \cite{ChrGun24}) is far more bizarre and less applicable when considering the setting of an adversary who wishes to transform a watermarked output into one that is still coherent and intelligible but does not possess a watermark.

The approaches of \cite{cgz24} and \cite{gm24} also consider outputs over a large token alphabet as we do. \cite{cgz24} considers a weaker form of robustness, \emph{substring completeness}, which requires that the watermark be detectable if the detection algorithm is given any unmodified substring of sufficient empirical entropy from the output; while they also consider reducing to a binary alphabet via fixed-length encoding in their construction of substring-complete watermarking, we note that this notion, unlike the stronger notion of robustness, is preserved in the translation to a larger token alphabet (as, clearly, an unmodified substring in the binary representation is likewise unmodified when translated to tokens). That said, a substring-complete watermarking scheme requires that a contiguous substring of the output be unmodified to be detectable, and, as the entropy requirement for such a substring is dependent on the security parameter, such a scheme cannot even provide robustness against any constant fraction of random substitutions.

Meanwhile, \cite{gm24} provides a very strong notion of \emph{edit distance robustness} against any constant fraction of even adversarial substitutions, deletions, or insertions; they do so by transforming a binary PRC robust to adversarial substitutions into a large-alphabet PRC with this improved notion of edit distance robustness, and applying a similar approach to that of \cite{ChrGun24} and our work to the transformed PRC. Specifically, their construction of an edit-distance-robust PRC works by \emph{indexing} the original PRC---roughly speaking, using an alphabet whose size is equal to the \emph{block length} of the original PRC, and letting each token in the expanded alphabet represent the presence of a 1 in its corresponding position in the respective block of the underlying PRC, so that insertions or deletions in the transformed PRC effectively become substitutions in the original. Of course, the drawback of this approach, as we have mentioned, is that the extremely large required alphabet size of the PRC leads to a correspondingly large entropy requirement (i.e., logarithmic in the security parameter) for the generative model's output, something which we explicitly seek to avoid in our results because natural language generally possesses large amounts of low-entropy tokens.

\paragraph{PRC constructions.} As mentioned above, \cite{ChrGun24} construct PRCs robust against adversarial substitutions or random substitutions and deletions, assuming the hardness of LPN. Under the same assumption, \cite{alrabiah2025ideal} showed the existence of PRC with stronger security guarantees, i.e., CCA security and adaptive robustness. \cite{ghentiyala2024new} construct PRCs against adversarial substitutions from weaker assumptions, as well as unconditionally-secure PRC with pseudorandomness against space-bounded adversaries. Most recently, \cite{cggm25} construct PRCs featuring subexponential security, a stronger notion of ``strongly adaptive robustness" which guarantees adaptive robustness even against an adversary with full knowledge of the watermarking key, and robustness to a constant fraction of (even arbitrary or adversarial) edits without the prohibitive entropy requirement of \cite{gm24}, albeit based on a non-standard ``permuted codes" assumption. In the realm of negative results, \cite{garg2025black,dackttling2025separating} recently showed that PRCs cannot be constructed from a wide-range of cryptographic primitives in a black-box manner.

\section{Definitions}

Let $\Nat = \braces{1, 2, 3, \ldots}$ and $[n] = \braces{m \in \Nat \mid m \leq n} = \braces{1, 2, \ldots, n}$. Take $1^n$ (given as an input to certain algorithms to represent the security parameter) to be the string of $n$ ones. Let $\textbf{1}_{P}$ be the \emph{indicator function} of some equality or inequality statement $P$, defined as being 1 when $P$ is true and 0 otherwise. Given a set $X$, we define $\Delta(X)$ to be the set of probability distributions over elements in $X$.

We say a statement is true for ``sufficiently large $n \in \Nat$" if there exists $N \in \Nat$ such that the statement holds for all $n \geq N$. We shall call a function $\epsilon(\cdot)$ \emph{negligible} if, for any polynomial $p(\cdot)$, $\epsilon(n) < 1/p(n)$ for sufficiently large $n \in N$; conversely, a function is \emph{non-negligible} or \emph{polynomial} if the above is not true. We refer to an algorithm as \emph{efficient} if its running time is polynomial in its input size.

Given a discrete distribution $D$, we define the (binary) \emph{entropy} of $D$ to be the quantity
\[H(D) = \sum_{x \in D} -D(x)\log_2 D(x)\]

Furthermore, given a selection $x \leftarrow D$ from a distribution, we define the \emph{empirical entropy} of $x$ in $D$ as the quantity $H_\text{emp}(x, D) = -\log_2 D(x)$.

\subsection{Error Channels}

Towards defining robustness for watermarking generative models whose outputs are large alphabets, we first must define generalized notions of error channels over families of large token alphabets. Formally:

\begin{definition}
Given some vocabulary of tokens $\V$, a \textbf{channel} is a (potentially randomized) algorithm $\E: \V^* \rightarrow \V^*$ that perturbs an output using tokens from $\V$.

Additionally, given a family of vocabularies $\V(n)$ parameterized by a security parameter $n \in \Nat$, we may consider a \textbf{family of channels} $\E = \braces{\E(n)}_{n \in \Nat}$ where, for each $n$, $\E(n)$ is a channel over vocabulary $\V(n)$.
\end{definition}

We next define standard notions of random substitution and deletion channels:

\begin{definition}
Given a family of vocabularies $\V(n)$ parameterized by a security parameter $n \in \Nat$, we define the \textbf{random deletion channel family} $\E_{\text{del}}(p) = \braces{\E(p, n)}_{n \in \Nat}$ where $\E(p, n)$, given constant $p \in [0,1]$ and security parameter $n$, takes as input a series of tokens $\tokens \in \V(n)^*$ and, for each $t \in \tokens$ in order, appends it to $\tokens'$ with probability $1-p$, then outputs $\tokens'$.
\end{definition}

\begin{definition}
Given a family of vocabularies $\V(n)$ parameterized by a security parameter $n \in \Nat$ and a family of functions $f = \braces{f_n: \V(n) \rightarrow \Delta(\V(n))}_{n \in \Nat}$, we define the \textbf{random substitution channel family} $\E_{\text{sub}}(f, p) = \braces{\E(f, p, n)}_{n \in \Nat}$ where $\E(f, p, n)$, given constant $p \in [0,1]$ and security parameter $n$, takes as input a series of tokens $\tokens \in \V(n)^*$ and outputs $\tokens' \in \V(n)^*$ where $t'_i = \braces{x :x \leftarrow f_n(t_i)}$ with probability $p$ and $t'_i = t_i$ otherwise.

For the case of a binary alphabet $\V = \Bit$, we also define the analogous \textbf{binary symmetric channel} $\E_{\text{bin}}(p)$, which takes as input a series of tokens $\tokens \in \Bit^*$ and outputs $\tokens' \in \Bit^*$ where $t'_i = 1 - t_i$ with probability $p$ and $t'_i = t_i$ otherwise.
\end{definition}

Notably, the definition of a random substitution channel for a non-binary vocabulary of tokens allows for instances of any token to be substituted with a (potentially arbitrarily determined) distribution over other tokens, capturing a variety of attacks such as simple paraphrasing attacks or randomized censorship that are practical when considering a large token alphabet but cannot be captured by random substitutions in a binary setting such as that of \cite{ChrGun24}.

Lastly, we define notions of arbitrary (potentially adversarial) channels in the binary and large-alphabet settings:

\begin{definition}
We refer to a channel $\E$ as a \textbf{$p$-bounded substitution channel} over a vocabulary $\V$ if, given some input $\tokens \in \V^*$, $\E$ always outputs $\tokens' \in \V^*$ where $\text{len}(\tokens') = \text{len}(\tokens)$ and, letting $\ell = \text{len}(\tokens)$, $t'_i \neq t_i$ for at most $p \ell$ indices $i \in [\ell]$.

If $\V = \Bit$, we refer to such a channel as a \textbf{$p$-bounded binary substitution channel}.
\end{definition}

\begin{definition}
Given strings $x$ and $x'$, let $D_\text{edit}(x, x')$ be the \emph{edit distance} between $x$ and $x'$, or the minimum number of insertions and deletions required to transform $x$ into $x'$ (or, symmetrically, vice versa). We refer to a channel $\E$ as a \textbf{$p$-bounded binary edit channel} over a vocabulary $\V$ if, given some input $\tokens \in \V^*$, $\E$ always outputs $\tokens' \in \V^*$ where $D_\text{edit}(\tokens, \tokens') \leq p \cdot \text{len}(\tokens)$.

If $\V = \Bit$, we refer to such a channel as a \textbf{$p$-bounded binary edit channel}.
\end{definition}

That is, a substitution channel is $p$-bounded if it makes at most a $p$ fraction of substitution edits to its input, and, respectively, a $p$-bounded edit channel makes at most a $p$ fraction of insertion or deletion edits, though these edits may be determined adversarially or arbitrarily.

\subsection{Generative Models and Watermarking}

We define generative models over large alphabets following the framework of \cite{gm24}, as computationally bounded algorithms which, given some input (prompt) and a sequence of tokens in a vocabulary $\V$ generated thus far, returns a probability distribution over $\V$ representing the next token to be generated.

\begin{definition}[\cite{gm24}]
A \textbf{generative model} is an efficient algorithm $\Generate$ over some vocabulary $\V$ that, given a prompt and sequence of tokens $\prompt, \vec{t} \in \V^*$, outputs a probability distribution $D \in \Delta(\V)$.

A \textbf{generative model family} is a family of generative models $\braces{\Generate(n)}_{n \in \Nat}$ parameterized by a security parameter $n$, such that $\Generate(n)$ operates on vocabulary $\V(n)$ where $|\V(n)|$ is bounded above by some polynomial $p(n)$, and so that there exists symbol $\End$ belonging to $\V(n)$ for every $n \in \Nat$ and polynomial $m(n)$ (the \textbf{maximum output length}) such that if $|\vec{t}| \geq m(n) - 1$ then $D \leftarrow \Generate(n)(\prompt, \vec{t})$ must have $D(\End) = 1$.
\end{definition}

To \emph{watermark} the output of a generative model, then, is to generate an output in such a way that it can be determined, using a secret key, that the output came from the watermarked model. Ideally, such a watermark preserves the distribution of the output from the model, is undetectable to parties not possessing the secret key, and is also robust to manipulation by an adversary attempting to remove the watermark. Formally:

\begin{definition}[\cite{gm24}]
Assume a family of generative models $\M = \braces{\Generate(n)}_{n \in \Nat}$ with vocabulary $\V(n)$ and maximum output length $m(n)$, and consider a tuple of efficient algorithms $(\Gen, \Watermark, \Detect)$, where:
\begin{itemize}
\item $\Gen(1^n) \rightarrow \sk$ takes as input a security parameter $n$ and outputs a secret key $\sk$,
\item $\Watermark_\sk(1^n, \prompt) \rightarrow \tokens$ takes as input the security parameter $n$, secret key $\sk$, and prompt $\prompt$, and outputs a sequence of tokens $\tokens \in \V(n)^{m(n)}$.\footnote{As in \cite{gm24}, if the model's output is shorter than the maximum length $m(n)$, we assume that the remainder of $\tokens$ is padded by repeating a special ``ending" token $\End$.}
\item $\Detect_\sk(1^n, \tokens) \rightarrow b$ takes as input the security parameter $n$, a secret key $\sk$, and a sequence of tokens $\tokens \in \V(n)^{m(n)}$, and outputs boolean $b \in \braces{\true, \false}$ corresponding to whether a watermark was detected in $\tokens$.
\end{itemize}

We refer to $(\Gen, \Watermark, \Detect)$ as a \textbf{secure watermarking scheme} for $\M$ if the following properties hold:
\begin{itemize}
\item \textbf{Soundness:} For any $\tokens \in \V(n)^{m(n)}$, there is negligible $\epsilon(\cdot)$ such that:
\[\Pr[\sk \leftarrow \Gen(1^n) : \Detect_\sk(1^n, \tokens) = \true] \leq \epsilon(n)\]
\item \textbf{Undetectability:} For any probabilistic polynomial-time oracle-aided adversary $\A(\cdot)$, there is negligible $\epsilon(\cdot)$ such that:
\[|\Pr[\sk \leftarrow \Gen(1^n) : \A^{\Watermark_\sk(1^n, \cdot)}(1^n) = 1] - \Pr[\A^{\O_\Generate(1^n, \cdot)}(1^n) = 1]| \leq \epsilon(n)\]

where $\Watermark$ is simply an oracle that outputs the result of $\Watermark$ given a prompt $\prompt$, and $\O_\Generate$ is an oracle that on input $\prompt$, generates a complete output from $\M$ by outputting the result of the following procedure:
\begin{enumerate}
\item Initialize $\tokens = ()$.
\item For $i \in [m(n)]$, generate $t_i \leftarrow \Generate(\prompt, \tokens)$ by drawing from the respective distribution output from $\Generate$, and append $t_i$ to $\tokens$. If $t_i = \End$ then pad the remainder of $\tokens$ with $\End$ and terminate, returning $\tokens$.
\end{enumerate}

\end{itemize}
Furthermore, for channel family $\E = \braces{\E(n)}_{n \in \Nat}$ where $\E(n) : \V(n)^* \rightarrow \V(n)^*$, we refer to $(\Gen, \Watermark, \\ \Detect)$ as \textbf{robust with respect to $\E$} if there is negligible $\epsilon(\cdot)$ such that for any prompt $\prompt$:
\[\Pr[\sk \leftarrow \Gen(1^n), \vec{t} \leftarrow \Watermark_\sk(1^n, \prompt) : \Detect_\sk(1^n, \E(n)(\vec{t})) = \false ] \leq \epsilon(n)\]

We may additionally restrict this definition to a set of prompts $\Pi$ (where robustness holds if the above is true for prompts $\prompt \in \Pi$).
\end{definition}

In order for an output to be undetectably and robustly watermarkable, it also must have some amount of entropy.\footnote{Consider for instance a prompt which instructs an LLM to output a determinstic word or string; it is trivial to see that such an output could not possibly be watermarked satisfactorily, since changing the output would violate undetectability while keeping it the same would violate robustness against any sort of editing.} While restricting the set of possible prompts is a way to achieve this for specific generative models, we consider for the general case a notion of ``entropy-restricted substring robustness" where we only require robustness to hold for outputs which have a sufficiently long substring containing a certain fraction of tokens with sufficiently high empirical entropy:

\begin{definition}
Assume a family of generative models $\M = \braces{\Generate(n)}_{n \in \Nat}$ with vocabulary $\V(n)$ and maximum output length $m(n)$, and let $(\Gen, \Watermark, \Detect)$ be a secure watermarking scheme for $\M$. Then, for a channel family $\E = \braces{\E(n)}_{n \in \Nat}$ where $\E(n) : \V(n)^* \rightarrow \V(n)^*$, we refer to $(\Gen, \Watermark, \Detect)$ as satisfying \textbf{$(L(n), \delta, h)$-entropy-restricted substring robustness with respect to $\E$} if there is negligible $\epsilon(\cdot)$ such that for any prompt $\prompt$ (which may be specified to be restricted to a set $\Pi$):
\begin{align*}
\Pr[\sk \leftarrow \Gen(1^n), \vec{t} \leftarrow \Watermark_\sk(1^n, \prompt) : & \exists j, k \in [\text{len}(\tokens)] : k - j \geq L(n)  \\ 
\text{and } \#\braces{i \in [j, k] : H_\text{emp}(t_i, \Generate&(\prompt, \tokens_{\leq i-1})) \geq h} \geq \delta (k - j)  \\
\text{ and } & \Detect_\sk(1^n, \E(\vec{t})) = \false  ]\leq \epsilon(n)
\end{align*}
where $\tokens_{\leq i-1}$ denotes the substring consisting of the first $i-1$ tokens in $\tokens$.
\end{definition}

That is, this definition requires that robustness holds only if there exists a substring of the output with length at least $L(n)$ where at least a $\delta$ fraction of the output tokens were generated with at least $h$ bits of empirical entropy with respect to the generative model's output distribution.

\subsection{Pseudorandom Error-Correcting Codes}
Closely related to watermarking generative algorithms is the concept of \emph{pseudorandom error-correcting codes}, introduced in \cite{ChrGun24}. These are a strengthening of the classical notion of error-correcting codes with the added requirement that the code for a given message should be \emph{pseudorandom}, or indistinguishable from uniform randomness by a computationally bounded adversary.

\begin{definition}[\cite{ChrGun24}]
Consider a tuple of efficient algorithms $\PRC = (\Gen, \Encode, \Decode)$, where:
\begin{itemize}
\item $\Gen(1^n) \rightarrow \sk$ takes as input a security parameter $n$ and outputs a secret key $\sk$,
\item $\Encode_\sk(m) \rightarrow c$ takes as input a secret key $\sk$ and message $m \in \V^k$ for some vocabulary $\V = \V(n)$ and input length $k = k(n)$, and outputs a code $c \in \V^\ell$ for some code length $\ell = \ell(n)$.
\item $\Decode_\sk(c) \rightarrow m$ takes as input a secret key $\sk$ and code $c \in \V^\ell$, and outputs either a message $m \in \V^k$ or $\bot$ to indicate that there is no valid decoding.
\end{itemize}
Given a channel $\E : \V^* \rightarrow \V^*$ that introduces some error into the code, we call $\PRC$  a \textbf{pseudorandom error-correcting code} (PRC) with robustness to $\E$ if the following properties are satisfied for any $n \in \Nat$ and the respective $\V(n)$, $k(n)$, $\ell(n)$:
\begin{itemize}
\item \textbf{Robustness (with respect to $\E$):} For any $m \in \V^k$, there is negligible $\epsilon(\cdot)$ such that:
\[\Pr[\sk \leftarrow \Gen(1^n), c \leftarrow \Encode_\sk(m) : \Decode_\sk(\E(c)) \neq m] \leq \epsilon(n)\]
\item \textbf{Soundness:} For any $c \in \V^*$, there is negligible $\epsilon(\cdot)$ such that:
\[\Pr[\sk \leftarrow \Gen(1^n) : \Decode_\sk(c) \neq \bot] \leq \epsilon(n)\]
\item \textbf{Pseudorandomness:} For any probabilistic polynomial-time oracle-aided adversary $\A(\cdot)$, there is negligible $\epsilon(\cdot)$ such that, letting $\O_\Encode(\sk)$ be an oracle that returns $\Encode_\sk(m)$ given message $m$ and $\O_\Unif$ be an oracle that implements a uniformly randomly selected function $r: \V^k \rightarrow\V^\ell$:
\[|\Pr[\sk \leftarrow \Gen(1^n) : \A^{\O_\Encode(\sk)}(1^n) = 1] - \Pr[\A^{\O_\Unif}(1^n) = 1]| \leq \epsilon(n)\]
\end{itemize}
\end{definition}

In this work we consider only \emph{zero-bit binary} PRCs, where $\V(n) = \Bit$ and $k = 0$---that is, there is only a single possible input (which by convention we label as $1$), and the code is a binary string of length $\ell(n)$.

%\TODO{do we want a different notion of hash functions than just a randomly selected function? pairwise-independent hash functions?}

\section{Protocol and Results}

Here, we formally restate the theorems presented in the introduction and present the protocol with which we prove them. To begin, we demonstrate a watermarking protocol which provides soundness, undetectability, and (given any output with ``minimal empirical entropy", or at least a constant fraction of tokens in a sufficiently long substring having a single bit of empirical entropy) robustness against any random substitution channel $\E_\text{sub} (f, 1/2 - \epsilon')$ for constant $\epsilon' > 0$, given the existence of a binary PRC with robustness to any (arbitrary) $(1/2 - \epsilon)$-bounded substitution channel for $\epsilon > 0$. Formally:

\begin{theorem}
\label{thm:substitution}
If, for any constant $\epsilon \in (0, 1/2]$, there exists a pseudorandom code $\PRC$ (having block length $\ell(\epsilon, n)$) with robustness against any $(1/2 - \epsilon)$-bounded binary substitution channel, then, for any generative model family $\M$ with token alphabet family $\V = \braces{\V(n)}_{n \in \Nat}$, family of functions $f = \braces{f_n: \V(n) \rightarrow \Delta(\V(n))}_{n \in \Nat}$, and constants $\epsilon' \in (0, 1/2]$ and $\delta \in (0, 1]$, there exists a secure watermarking scheme for $\M$, $(\Gen, \Watermark, \Detect)$, that satisfies $(n \ell(\epsilon'\delta/17, n), \delta, 1)$-entropy-restricted substring robustness with respect to the random substitution channel $\E_{\text{sub}}(f, 1/2 - \epsilon')$.
\end{theorem}

Combining this with the respective PRC construction from \cite{ChrGun24} gives:

\begin{corollary}
\label{cor:substitution}
If Assumption 1 of \cite{ChrGun24} is true\footnote{This assumption, at a high level, requires either subexponential hardness of LPN or polynomial hardness of both LPN and a low-density planted XOR problem.}, then, letting $\ell(\epsilon, n)$ be the block length of the respective PRC (in Theorem 1 and Theorem 2 of \cite{ChrGun24}) required to achieve robustness against any $(1/2 - \epsilon)$-bounded binary substitution channel, for any generative model family $\M$ with token alphabet family $\V = \braces{\V(n)}_{n \in \Nat}$, family of functions $f = \braces{f_n: \V(n) \rightarrow \Delta(\V(n))}_{n \in \Nat}$, and constants $\epsilon' \in (0, 1/2]$ and $\delta \in (0, 1]$, there exists a secure watermarking scheme $(\Gen, \Watermark, \Detect)$ for $\M$ that satisfies $(n \ell(\epsilon'\delta/17, n), \delta, 1)$-entropy-restricted substring robustness with respect to the random substitution channel $\E_{\text{sub}}(f, 1/2 - \epsilon')$.
\end{corollary}

Next, we would like to extend the above result to demonstrate that $\Watermark$ is additionally robust to a large fraction of random deletions if the same also holds for $\PRC$. However, the construction from \cite{ChrGun24} of a PRC robust to  substitutions and deletions only provides robustness against \emph{random} substitutions---specifically, substitutions made in the pattern of a binary symmetric channel---whereas we recall from the above that the errors introduced by $\Watermark$, while bounded by Claim \ref{clm:expect}, are not necessarily random. 

Thus, in order to leverage this construction, we require an additional assumption that the errors introduced by $\Watermark$ are random (described in \cite{ChrGun24} as the text being ``sufficiently variable"), to demonstrate that $\Watermark$ is robust to a large fraction of random substitutions and deletions. This is somewhat more involved in our case, as the definition of a random substitution is different from that given in the binary model of \cite{ChrGun24}. Random substitutions are equivalent to a binary symmetric channel in the case they consider, since therein the generative model's output is binary and directly corresponds to the PRC codeword they attempt to embed; however, because we require the PRC codeword $c$ to instead correspond to the \emph{image} of the output under the hash function $h$, random substitutions made to the output do not necessarily result in substitutions to the PRC codeword, and, in fact, whether they do may be dependent on the function $h$ chosen by $\Watermark$.

So, while in the case of \cite{ChrGun24} it is sufficient to assume that the embedding process itself produces errors whose distribution is given by a binary symmetric channel in order to leverage robustness against random substitutions (as the composition of that with the binary symmetric channel of random substitutions is itself a binary symmetric channel), we need to assume that the composition of the embedding process and the random substitution channel produces errors in a distribution modelable by a binary symmetric channel, as given in Definition \ref{def:consistent}. Formally:

\begin{definition}
\label{def:consistent}
We say that a prompt $\prompt$ \textbf{produces consistently distributed errors} for a watermarking scheme $(\Gen, \Watermark, \Detect)$ with respect to a substitution channel family $\E$ if there exists some constant $p$ such that, taken over the randomness of $\Watermark$ and the channel $\E$, each token $t$ of the output, after being generated by $\Watermark$ (given prompt $\prompt$) according to some function $h$ and bit $c_i$ of codeword $c$ from $\PRC$ and then modified according to $\E$, has an identical probability $p$, independent from that of other tokens, of having $h(t) \neq c_i$.
\end{definition}

Given this, we can present our following result that leverages this definition:

\begin{theorem}
\label{thm:subdel1}
If, for any constants $\epsilon_{sub} \in (0, 1/2]$ and $\epsilon_{del} \in (0, 1]$, there exists a pseudorandom code $\PRC$ (having block length $\ell(\epsilon_{sub}, \epsilon_{del}, n)$) with robustness against the composition of channels $\E_\text{bin}(\frac{1}{2} - \epsilon_{sub}) \circ \E_\text{del} (1 - \epsilon_{del})$, then, for any generative model family $\M$ with token alphabet family $\V = \braces{\V(n)}_{n \in \Nat}$, family of functions $f = \braces{f_n: \V(n) \rightarrow \Delta(\V(n))}_{n \in \Nat}$, and constants $\epsilon'_{sub} \in (0, 1/2]$, $\epsilon'_{del} \in (0, 1]$, and $\delta \in (0, 1]$, there exists a secure watermarking scheme $(\Gen, \Watermark, \Detect)$ for $\M$ that, for any family $\Pi$ of prompts where any $\prompt \in \Pi$ produces consistently generated errors for $(\Gen, \Watermark, \Detect)$ with respect to the substitution channel $\E_{\text{sub}}(f, 1/2 - \epsilon')$, satisfies $(n \ell(\epsilon'_{sub} \delta / 17, \epsilon'_{del}, n), \delta, 1)$-entropy-restricted substring robustness over $\Pi$ with respect to the composition of channels $\E_{\text{sub}}(f, 1/2 - \epsilon'_{sub}) \circ \E_\text{del} (1 - \epsilon'_{del})$.
\end{theorem}

Again, combining this with the respective PRC construction from \cite{ChrGun24} gives:

\begin{corollary}
\label{cor:subdel1}
If Assumption 1 of \cite{ChrGun24} is true, then, letting $\ell(\epsilon_{sub}, \epsilon_{del}, n)$ be the block length of the PRC given by Theorem 4 of \cite{ChrGun24} required to achieve robustness against the composition of channels $\E_\text{bin}(\frac{1}{2} - \epsilon_{sub}) \circ \E_\text{del} (1 - \epsilon_{del})$, for any generative model family $\M$ with token alphabet family $\V = \braces{\V(n)}_{n \in \Nat}$, family of functions $f = \braces{f_n: \V(n) \rightarrow \Delta(\V(n))}_{n \in \Nat}$, and constants $\epsilon'_{sub} \in (0, 1/2]$, $\epsilon'_{del} \in (0, 1]$, and $\delta \in (0, 1]$, there exists a secure watermarking scheme $(\Gen, \Watermark, \Detect)$ for $\M$ that, for any family $\Pi$ of prompts where any $\prompt \in \Pi$ produces consistently generated errors for $(\Gen, \Watermark, \Detect)$ with respect to the substitution channel $\E_{\text{sub}}(f, 1/2 - \epsilon')$, satisfies $(n \ell(\epsilon'_{sub} \delta / 17, \epsilon'_{del}, n), \delta, 1)$-entropy-restricted substring robustness over $\Pi$ with respect to the composition of channels $\E_{\text{sub}}(f, 1/2 - \epsilon'_{sub}) \circ \E_\text{del} (1 - \epsilon'_{del})$.
\end{corollary}

Alternatively, though none is currently known, if there were a construction of $\PRC$ that guaranteed robustness against \emph{adversarial} substitutions and random deletions, one could bypass the need for heuristic assumptions such as that given by Definition \ref{def:consistent} and instead leverage Claim \ref{clm:robust} directly to demonstrate that $\Watermark$ likewise satisfies robustness to random deletions, as follows:

\begin{theorem}
\label{thm:subdel2}
If, for any constants $\epsilon_{sub} \in (0, 1/2]$ and $\epsilon_{del} \in (0, 1]$, there exists a pseudorandom code $\PRC$ (having block length $\ell(\epsilon_{sub}, \epsilon_{del}, n)$) with robustness against the composition of channels $\E \circ \E_\text{del} (1 - \epsilon_{del})$ for any $(1/2 - \epsilon_{sub})$-bounded binary substitution channel $\E$, then, for any generative model family $\M$ with token alphabet family $\V = \braces{\V(n)}_{n \in \Nat}$, family of functions $f = \braces{f_n: \V(n) \rightarrow \Delta(\V(n))}_{n \in \Nat}$, and constants $\epsilon'_{sub} \in (0, 1/2]$, $\epsilon'_{del} \in (0, 1]$, and $\delta \in (0, 1]$, there exists a secure watermarking scheme $(\Gen, \Watermark, \Detect)$ for $\M$ that satisfies $(n \ell(\epsilon'_{sub} \delta / 17, \epsilon'_{del}, n), \delta, 1)$-entropy-restricted substring robustness over $\Pi$ with respect to the composition of channels $\E_{\text{sub}}(f, 1/2 - \epsilon'_{sub}) \circ \E_\text{del} (1 - \epsilon'_{del})$.
\end{theorem}

Lastly, we show that we can leverage a PRC such as the new construction of \cite{cggm25} based on a ``permuted codes conjecture" applied to folded Reed-Solomon codes, which features robustness to both a large (close to $1/2$) fraction of (potentially arbitrary) substitutions and a very small but constant fraction of (potentially arbitrary) edits, to construct low-entropy watermarking robust to a constant fraction of arbitrary edits. However, this comes at a cost to the length of the minimal-entropy substring required to prove robustness; in particular, for the case where we allow adversarial or arbitrary edits, we require \emph{a constant fraction of the entire output} to have sufficient entropy, as otherwise it is fairly straightforward to see that a constant fraction of arbitrary edits could specifically target the minimal-entropy substring for corruption and thus break robustness. Thus, we conclude:

\begin{theorem}
\label{thm:edit}
If, for any constant $\epsilon_{sub} \in (0, 1/2]$, there exists a pseudorandom code $\PRC$ and a constant $\epsilon_{edit} \in (0, 1]$ such that $\PRC$ has robustness against the composition of channels $\E \circ \E'$ for any $(1/2 - \epsilon_{sub})$-bounded binary substitution channel $\E$ and any $\epsilon_{edit}$-bounded binary edit channel $\E'$, then, for any generative model family $\M$ (having maximum output length $m(n)$) and constants $\delta, \alpha \in (0, 1]$, there exists a secure watermarking scheme $(\Gen, \Watermark, \Detect)$ for $\M$ and constant $\epsilon'_{edit} \in (0, 1]$ such that $(\Gen, \Watermark, \Detect)$ satisfies $(\alpha m(n), \delta, 1)$-entropy-restricted substring robustness with respect to any $\epsilon'_{edit}$-bounded edit channel over the vocabulary of $\M$.
\end{theorem}

Combining this with the respective PRC construction from \cite{cggm25} gives:

\begin{corollary}
\label{cor:edit}
If Conjecture 3.1 of \cite{cggm25} is true, then, for any generative model family $\M$ (having maximum output length $m(n)$) and constants $\delta, \alpha \in (0, 1]$, there exists a secure watermarking scheme $(\Gen, \Watermark, \Detect)$ for $\M$ and a constant $\epsilon'_{edit} \in (0, 1]$ such that $(\Gen, \Watermark, \Detect)$ satisfies $(\alpha m(n), \delta, 1)$-entropy-restricted substring robustness with respect to any $\epsilon'_{edit}$-bounded edit channel over the vocabulary of $\M$.
\end{corollary}

We provide greater specificity on the constant $\epsilon'_{edit}$ alongside the proof in Section \ref{sec:edit}.
The watermarking scheme $(\Gen, \Watermark, \Detect)$ we use to prove the above results is presented in Figures \ref{fig:gen} ($\Gen$), \ref{fig:watermark} ($\Watermark$), and \ref{fig:detect} ($\Detect$); the proofs are presented in the subsequent section.

\protocol{}{Protocol $\Gen$ for generating the secret information used in the watermarking protocol.}{fig:gen}{
{\bf Input:} Security parameter $n$, generative model $\Generate = \Generate(n)$ over token vocabulary $\V = \V(n)$ with maximum output length $m = m(n)$ (belonging to family $\braces{\Generate(n)}_{n \in \Nat}$), prompt $\prompt$, error-correcting pseudorandom code $\PRC$ with code length $\ell = \ell(n)$ such that $m(n) = \omega(n \ell(n))$. \\
{\bf Output:} Secret key $\sk$ for $\PRC$ and a tuple of functions $\vec{h} \in \H^{\ceil{m/\ell}}$, where $\H$ is the set of functions $h: \V \rightarrow \Bit$.
\medskip

{\bf Protocol:} 
\begin{enumerate}
\item Generate $\sk \leftarrow \PRC.\Gen(1^n)$.

\item Generate $\vec{h}$ by, for $i \in [\ceil{m/\ell}]$, randomly sampling $h_i \leftarrow \H$.

\item Return $\sk$ and $\vec{h}$.

\end{enumerate}
}

\protocol{}{Protocol $\Watermark$ for watermarking a generative model.}{fig:watermark}{
{\bf Input:} Security parameter $n$, generative model $\Generate = \Generate(n)$ over token vocabulary $\V = \V(n)$ with maximum output length $m = m(n)$ (belonging to family $\braces{\Generate(n)}_{n \in \Nat}$), prompt $\prompt$, error-correcting pseudorandom code $\PRC$ with code length $\ell = \ell(n)$, secret key $\sk$ for $\PRC$, and tuple of functions $\vec{h} \in \H^{\ceil{m/\ell}}$, where $\H$ is the set of all functions $h: \V \rightarrow \Bit$. \\
{\bf Output:} Watermarked output $\tokens \in \V^*$.
\medskip

{\bf Protocol:} 
\begin{enumerate}
\item Initialize $\tokens = ()$, $i = 1$, and $j = 1$.

\item Draw $c \leftarrow \PRC.\Encode_\sk(1^n)$. Let $h = h_j$, the $j^\text{th}$ element of $\vec{h}$.

\item Let $D = \Generate(\prompt, \tokens)$.

\item Randomly and independently sample two tokens $t_0, t_1 \leftarrow D$. Generate the next token $t$ by selecting one of these two as follows:
\begin{enumerate}
\item If $h(t_0) = h(t_1)$, let $t = t_b$ for a uniformly random $b \leftarrow \Bit$.
\item Otherwise, if $h(t_0) \neq h(t_1)$, let $t$ be whichever of the two samples $t_b$ has $h(t_b) = c_i$ (where $c_i$ is the $i^\text{th}$ bit of $c$). That is, if $h(t_0) = 0$, let $t = t_{c_i}$, and otherwise let $t = t_{1 - c_i}$.
\end{enumerate}

\item Append $t$ to $\tokens$, then:
\begin{enumerate}
\item If $t = \End$, terminate (padding the remainder of $\tokens$ to length $m(n)$ with $\End$ tokens if necessary) and output $\tokens$.
\item Otherwise, if $i = \ell$, start a new block by returning to step (2), incrementing $j$, and resetting $i$ to 1.
\item Otherwise, generate the next token of the current block by incrementing $i$ and returning to step (3).
\end{enumerate}
\end{enumerate}
}

\protocol{}{Protocol $\Detect$ for detecting the presence of a watermark in a given message.}{fig:detect}{
{\bf Input:} Security parameter $n$, message $\tokens \in \V^*$ (for $\V = \V(n)$), error-correcting pseudorandom code $\PRC$ with code length $\ell(n)$, secret key $\sk$ for $\PRC$, and tuple of functions $\vec{h} \in \H^{\ceil{m(n)/\ell(n)}}$, where $\H$ is the set of functions $h: \V \rightarrow \Bit$. \\
{\bf Output:} Boolean $b$ indicating whether a watermark was detected.
\medskip

{\bf Protocol:} 
\begin{enumerate}
\item Let $\mathsf{len} = \min(|\tokens|, m(n))$. For $i \in [\mathsf{len}]$, $j \in [i, \min(i+\ell(n)-1, \mathsf{len})]$, $k \in |\vec{h}|$:
\begin{enumerate}
\item If $\PRC.\Decode_\sk(h_k(\tokens_i), \ldots, h_k(\tokens_j)) \neq \bot$, return $\true$.
\end{enumerate}
\item If none of the above calls to $\PRC.\Decode$ succeed, return $\false$.
\end{enumerate}
}

\section{Proof}

We begin by proving that $(\Gen, \Watermark, \Detect)$ is a secure watermarking scheme (i.e., undetectable and sound), as these properties hold regardless of the robustness of $\PRC$ (though they do require pseudorandomness and soundness); afterwards, we prove the definitions of robustness for each of our results in turn based on the respective required robustness properties for $\PRC$.

\subsection{Undetectability and Soundness}

\begin{claim*}
\label{clm:sound}
$(\Gen, \Watermark, \Detect)$ satisfies soundness and undetectability for any family of generative models $\M = \braces{\Generate(n)}_{n \in \Nat}$.
\end{claim*}

\begin{proof}
Both properties are fairly straightforward to demonstrate. 

Soundness follows from the soundness of $\PRC$; notice that $\Detect$, on any given input, calls $\PRC.\Decode$ at most a number of times polynomial in the security parameter $n$, since $m(n)$, $\ell(n)$, and the size of $\vec{h}$ are polynomial in $n$. Since $\sk$ is generated by $\PRC.\Gen$, the probability of each such call returning a result besides $\bot$ is negligible in $n$, and thus the probability over $\sk$ and $\vec{h}$ of $\Detect$ returning $\true$ must likewise be negligible by a union bound.

Undetectability follows from the pseudorandomness of $\PRC$. First consider an ``ideal'' world where the code $c$ from $\PRC$ is replaced with true randomness, and thus each $c_i$ is $0$ or $1$ with exactly probability $1/2$. In this world we can show easily that $\Watermark$ perfectly preserves the distribution $D$ of each token generated by $\Generate$. To see this, consider any token $t^* \in \V$. If $t^*$ is drawn as $t_0$ and $t_1 \neq t^*$ (which happens with probability $D(t^*)(1 - D(t^*))$), and $p$ is the probability that $h(x) \neq h(t^*)$ conditioned on $x \leftarrow D$ and $x \neq t^*$, the probability that $t^*$ is selected as the next token $t$ is:

\[p \Pr[c_i = h(t^*)] + (1-p) \frac{1}{2} = \frac{1}{2}p + \frac{1}{2}(1-p) = \frac{1}{2}\]
since, in $\Watermark$, $t_0$ is selected when $c_i = h(t_0)$ (which, as we assume $c_i$ to be uniformly random, happens with probability exactly 1/2) if $h(t_0) \neq h(t_1)$ and with probability 1/2 otherwise. The same is true symmetrically when $t^*$ is drawn as $t_1$ and $t_0 \neq t^*$, and when $t^*$ is drawn as \emph{both} $t_0$ and $t_1$ (which happens with probability $D(t^*)^2$) then it is selected with probability 1. Thus, as all of these events are mutually exclusive, the probability that $t^*$ is selected as the next token $t$ is exactly:

\[2  (D(t^*)(1 - D(t^*))\bigg(\frac{1}{2} \bigg) + D(t^*)^2 = D(t^*)\]

So if $c$ is truly random, it is straightforward to see that, in each iteration of $\Watermark$, the next token $t$ is drawn with exactly the same probability as it would be by $\Generate$, and thus the output of $\Watermark$ is identically distributed to the output of $\Generate$ on any specific prompt $\prompt$. We shall denote the above ``ideal'' experiment by $\Ideal$.

However, as $c$ is instead the \emph{pseudorandom} output of $\PRC$, assume towards contradiction that there is some adversary $\A$ and polynomial $p(\cdot)$ where
\[|\Pr[\sk \leftarrow \Gen(1^n) : \A^{\Watermark_\sk(1^n, \cdot)}(1^n) = 1] - \Pr[\A^{\O_\Generate(1^n, \cdot)}(1^n) = 1]| \geq p(n)\]

We leverage $\A$ to construct an adversary $\A'$ which contradicts the pseudorandomness of $\PRC$. $\A'$ runs $\A$, but every time $\A$ invokes its oracle $\A'$ runs and returns the result of running $\Watermark$ on the given prompt with the single exception that every code $c$ generated from $\PRC$ is replaced by the output of a call to its own oracle.

If the oracle given to $\A'$ calls $\PRC.\Encode$, then the experiment $\A'$ runs is identical to $\Watermark$; if it instead generates uniform randomness, then the experiment is instead identical to $\Ideal$, whose output is as we have already argued identically distributed to generating tokens with $\Generate$. Thus,
\[|\Pr[\sk \leftarrow \Gen(1^n) : \A'^{\O_\Encode(\sk)}(1^n) = 1] - \Pr[\A'^{\O_\Unif}(1^n) = 1]|\]
is by construction equal to
\[|\Pr[\sk \leftarrow \Gen(1^n) : \A^{\Watermark_\sk(1^n, \cdot)}(1^n) = 1] - \Pr[\A^{\O_\Generate(1^n, \cdot)}(1^n) = 1]|\]
But since we assumed this latter quantity was at least polynomial $p(n)$, the former must be as well, contradicting the pseudorandomness of $\PRC$ and completing the argument that $(\Gen, \Watermark, \Detect)$ satisfies undetectability.

\end{proof}

\subsection{Robustness}

The robustness guarantee of our watermarking scheme depends primarily on the robustness of $\PRC$, since, while the scheme itself generates errors attempting to embed the PRC codeword $c$ into (the hash of) the output, we show that these errors are comprised of a limited number of specifically substitution errors, which can be bounded as long as sufficiently many of the output tokens are generated with enough empirical entropy. Notably, however, these errors are not necessarily random, as, depending on the behavior of the model, tokens could be generated in such a way that the errors are correlated with significant probability over the choice of $h$. 

\subsubsection{Random Substitutions}

If we restrict to the case of only random substitutions, we note that an alternative version of $\Watermark$ that generates a new function $h$ for every token, as opposed to reusing the same $h$ for an entire codeword of $\PRC$, would be sufficient to prove robustness, since without insertions or deletions the number and order of the tokens cannot change (and thus one could always track which function $h$ was used to generate each token, even after random substitutions were made). However, this alternative approach fails with even a single insertion or deletion, and additionally only gains a small constant factor in the analysis for random substitutions, so we opt to instead present and analyze the existing version of $\Watermark$ to better set up the subsequent results for robustness against substitutions and deletions.

For the case of random substitutions only, we demonstrate the following claim:

\begin{claim*}
\label{clm:robust}
For constants $\epsilon, \epsilon', \delta > 0$ such that $\epsilon < \epsilon' \delta / 16$, if $\PRC$ having block length $\ell(n)$ is robust against any $(1/2 - \epsilon)$-bounded binary substitution channel, then, for any family $f = \braces{f_n: \V(n) \rightarrow \Delta(\V(n))}_{n \in \Nat}$, $\Watermark$ satisfies $(n \ell(n), \delta, 1)$-entropy-restricted substring robustness against the random substitution channel family $\E_\text{sub}(f, 1/2-\epsilon')$.
\end{claim*}

(We note that \cite{ChrGun24} demonstrates constructions of the requisite PRC.)

\begin{proof}

First, similarly to the other two properties, it suffices to prove robustness by bounding the number of errors (i.e., tokens $t_i$ where $h_j(t_i) \neq c_i$) in the ``ideal" experiment described in Claim \ref{clm:sound} where the codewords $c$ generated by $\PRC$ are replaced by uniformly random bits; as shown in Claim \ref{clm:sound}, pseudorandomness of $\PRC$ guarantees that the view of this experiment is identically distributed to that of the real execution of $\Watermark$ except with negligible probability, and thus, if in the ideal experiment there exists with overwhelming probability a block of the random codeword $c$ where sufficiently many tokens in the output are correctly embedded, it follows that there will also with overwhelming probability be such a block of the pseudorandom code $\PRC$ (which will thus be recoverable) in the real experiment.

We next need to bound the error rate of the output of $\Watermark$ introduced by the embedding process. Intuitively, in each step, the closer the distribution $\braces{h(d):d \leftarrow D}$ is to uniform over $\Bit$, the more likely the algorithm is to draw a pair $(t_0, t_1)$ such that $h(t_0) \neq h(t_1)$ and thus output a token that hashes to the correct value of $c_i$ from $\PRC$ without biasing the distribution of tokens. We make the following claim to show that, in a block where a sufficient number of tokens are drawn with enough empirical entropy, it is likely that enough of those tokens are drawn from sufficiently ``close to uniform" distributions so that enough of \emph{those} tokens are able to be drawn with hash value matching the PRC codeword to make the watermark easily detectable: 

\begin{claim*}
\label{clm:expect}
Consider the execution of $\Watermark$ over some block of $\PRC$ (i.e., the iteration from $i = 1$ to $i = \ell(n)$ for some fixed $j$). There exists negligible $\nu(\cdot)$ such that, with at least probability $1/4$ over the choice of function $h$, the fraction of tokens $t_i$ for which $h_j(t_i) \neq c_i$ is, except with probability $\nu(n)$ over the randomness of $\Watermark$, at most $1/2 - \delta' / 16$, where $\delta'$ is the fraction of tokens drawn with empirical entropy at least 1.
\end{claim*}

\begin{proof}
Consider an alternate experiment to generate the given block of $\Watermark$ where, to generate every token $t_i$, we first sample a \emph{new} function $h \leftarrow \H$, and sample $t_i$ using the respective function $h$ but otherwise according to $\Watermark$. Then, given some token $t_i$ drawn with empirical entropy at least 1, this means that it was sampled from a distribution $D$ such that $D(t_i) \leq 1/2$. In particular, this means that the probability that $t_i$ was generated from a pair $(t_0, t_1)$ such that $t_0 = t_1$ must likewise be at most 1/2. If $t_0 \neq t_1$, then the probability over the randomly generated function $h$ that $h(t_0) \neq h(t_1)$ is exactly 1/2 (meanwhile, if $t_0 = t_1$, this probability is clearly 0); thus, the probability that $h(t_0) \neq h(t_1)$ during the generation of this specific token $t_i$ must be at least 1/4.

Notice that if $h(t_0) \neq h(t_1)$, then by construction of $\Watermark$ we will always select $t_i$ such that $h(t_i) = c_i$---that is, the embedding of the bit $c_i$ of the PRC codeword will be successful. Meanwhile, if $h(t_0) = h(t_1)$, then the embedding will be successful if and only if $h(t_0) = h(t_1) = c_i$, which over randomly generated $h$ happens with probability exactly 1/2. 

\iffalse
Combining this with the above yields that, if $t_i$ has empirical entropy of at least 1, then $h(t_i) = c_i$ must be true with probability at least $1/4 (1) + 3/4 (1/2) = 5/8$.
\fi

In the actual execution of $\Watermark$, $h$ remains the same throughout the entire generation of the block; however, because tokens are always generated according to their unbiased probabilities from $\Generate(\prompt, \\ \vec{t})$ by construction, we note that the tokens generated are identically distributed to the above experiment (or to the output of simply generating non-watermarked tokens using $\Generate(\prompt, \vec{t})$). Because of this, it still holds that, if we sample two independently random tokens $t_0$ and $t_1$ as above, there is still probability $1/4$ over the choice of $h$ that $h(t_0) \neq h(t_1)$ for every token $t_i$ generated with sufficient empirical entropy. Unlike the above experiment, though, these probabilities are not independent as tokens may be repeated, so it is not necessarily the case that the fraction of tokens satisfying the entropy requirement with $h(t_0) \neq h(t_1)$ is close to $1/4$ except with negligible probability; however, we may still conclude by an averaging argument that, with at least probability $1/4$ over the choice of $h$, at least $1/4$ of the tokens with sufficient empirical entropy have $h(t_0) \neq h(t_1)$.\footnote{This may seem unintuitive since, at first glance, one would expect the tokens selected, and thus the indices of the tokens with sufficient experimental entropy, to be affected by the choice of $h$; however, we again stress that, because $\Watermark$ perfectly preserves the distributions of tokens generated, the tokens are independent of $h$ and $c$. In particular, one can imagine any set $S \subset [\ell(n)]$ of tokens and argue that, conditioned on the tokens with indices in $S$ having sufficient empirical entropy, at least $1/4$ of the tokens with indices in $S$ have $h(t_0) \neq h(t_1)$ with probability $1/4$ over randomly generated $h$.}

Thus, in the case (which occurs with at least probability $1/4$) where at least $1/4$ of the tokens with sufficient empirical entropy (i.e., a $\delta' / 4$ fraction of total tokens) have $h(t_0) \neq h(t_1)$, those tokens will \emph{never} have $h_j(t_i) \neq c_i$, and the remainder will have it with at most probability $1/2$. This means that, \emph{in expectation}, at most a $1/2 (1 - \delta'/4) = 1/2 - \delta'/8$ fraction of tokens can have $h_j(t_i) \neq c_i$. However, for the tokens with $h(t_0) = h(t_1)$, the probabilities that $h_j(t_i) \neq c_i$ are independent in the ``ideal" experiment where the codeword $c$ is randomly generated, and thus it can be asserted that, except with negligible probability, the fraction of those tokens for which $h_j(t_i) \neq c_i$ is actually true will be close to its expectation.

Specifically, assume at least a $1/4$ fraction of tokens have $h(t_0) = h(t_1)$; otherwise we are already done (as then a $3/4$ fraction would be guaranteed to have $h_j(t_i) = c_i$). Then, by Hoeffding's inequality, the probability that the fraction of all tokens with both $h(t_0) = h(t_1)$ and $h_j(t_i) \neq c_i$ is at least $\delta'/16$ greater than its expectation---in other words, that the \emph{number} of such tokens is at least $\delta' \ell(n)/16$ greater than its expectation---is at most

\[\text{exp} \bigg( \frac{-2(\delta' \ell(n)/16)^2}{ \ell(n)/4} \bigg) = \text{exp}(- \delta'^2 \ell(n) / 32)\]

which is negligible in $n$ because $\delta'^2/32$ is a (small) constant and $\ell(\cdot)$ is polynomial. And, in the case where a $\delta' / 4$ fraction of total tokens have $h(t_0) \neq h(t_1)$ but the above event is \emph{not} true (which, letting $\nu(n)$ be the negligible function above, must happen with at least probability $1/4 - \nu(n)$), the actual fraction of tokens with $h_j(t_i) \neq c_i$ can be at most

\[1/2 - \delta'/8 + \delta'/16 = 1/2 - \delta'/16\]

which completes the argument.

\end{proof}

Given this, assume that an output exists some substring of length at least $n \ell(n)$ containing at least a $\delta$ fraction of tokens with empirical entropy 1 (otherwise, if no such substring exists entropy-restricted substring robustness holds trivially). In this case, this substring must contain $k \geq n-1$ complete blocks of the PRC, and those blocks must contain at least a $\delta - \frac{1}{n}$ fraction of the tokens with sufficient entropy, which we may take to be greater than $\delta/2$ for sufficiently large $n > 2/\delta$. 

Then, by an averaging argument, at least $\delta k / 2$ blocks in that substring have at least a $\delta/4$ fraction of such tokens, which by Claim \ref{clm:expect} means that those blocks, with probability $1/4$ over the choice of functions $h$, have (with overwhelming probability) at most a $1/2 - \delta / 64$ error rate for embedding $c$. We shall refer to such blocks as ``good" for later analysis. Notably, because an independently random function $h$ is drawn for each block, these probabilities are independent.

Now, let us assume that the output is subjected to the random substitution channel $\E_\text{sub}(f, 1/2-\epsilon')$---that is, with independent probability $1/2-\epsilon'$, each token is replaced with a different token determined according to an arbitrary function $f$.\footnote{Namely, consider mapping each token in the alphabet to a distribution of replacement tokens and, for each replacement, selecting a token from the corresponding distribution, as long as the same distribution is used for each instance of a token.} Given this, we want to show that with overwhelming probability there is at least one block of the output that still has at most a $(1/2 - \epsilon)$ fraction of tokens $t_i$ for which $h_j(t_i) \neq c_i$; if we can conclude this, then by the robustness of $\PRC$ the watermark must be detectable with overwhelming probability by hashing that block and running $\PRC.\Decode$ during the execution of $\Detect$. Formally, we prove this via the following:

\begin{claim*}
\label{clm:good}
Given some ``good" block of the output of $\Watermark$, after $\E_\text{sub}(f, 1/2-\epsilon')$ is applied the expected fraction of tokens $t_i$ in that block for which $h_j(t_i) \neq c_i$ is no more than $1/2 - \epsilon' \delta / 16$.
\end{claim*}

\begin{proof}
Assume as a worst case that, whenever a substitution is made to a token which, in the experiment given in Claim \ref{clm:expect}, was drawn with $h(t_0) \neq h(t_1)$ (i.e., selected from a pair whose elements hashed to different values), that substitution \emph{always} changes the hash value of $h$ applied to the token, creating a guaranteed error. Other tokens are errors with probability $1/2$, meaning that in expectation they will retain that probability after being substituted irrespective of the probability with which substitution changes the value of $h$.

In this case, let $\psi$ be the fraction of tokens with $h(t_0) \neq h(t_1)$. Then in expectation the resulting fraction of errors is
\[\psi \bigg(\frac{1}{2} - \epsilon' \bigg) + (1 - \psi) \frac{1}{2} = \frac{1}{2} - \epsilon' \psi\]

For ``good" blocks we assumed that $\psi$ was at least $\delta/8$ (since Claim \ref{clm:expect} selects blocks where more than 1/4 of the tokens with empirical entropy 1 have $h(t_0) \neq h(t_1)$ and we apply it to blocks with at least a $\delta/4$ fraction of tokens having empirical entropy 1), thus the conclusion follows.
\end{proof}

Markov's inequality then gives that the probability that, in a ``good" block, the fraction of tokens $t_i$ for which $h_j(t_i) \neq c_i$ is greater than $(1/2 - \epsilon)$ (i.e., the watermark cannot be successfully detected) can be at most
\[\frac{1/2 - \epsilon' \delta / 16}{1/2 - \epsilon}\]
which in particular means that, as long as $\epsilon < \epsilon' \delta / 16$, the probability that a ``good" block is recoverable is a (small) constant. Importantly, the probability of this occurrence is independent between blocks due to the resampling of $h$, which means that, since there are at least $\delta k/2$ blocks having enough tokens with empirical entropy 1, each of which has an (independent) $1/4$ probability to be ``good" and the above constant (and also independent) probability to, if ``good", have few enough tokens $t_i$ for which $h_j(t_i) \neq c_i$ that robustness of $\PRC$ guarantees that the codeword $c$ can be successfully detected, we can conclude that the probability of $\Detect$ failing to detect the watermark is negligible (inverse exponential) in $k \geq n-1$ and thus in $n$, as desired.
    
\end{proof}

Together with Claim \ref{clm:sound}, Claim \ref{clm:robust} implies Theorem \ref{thm:substitution} and thus Corollary \ref{cor:substitution}.

\subsubsection{Random Substitutions and Deletions}

Continuing to the case of random substitutions and deletions, we present the following claims, which are both closely related to (and whose proofs largely follow from the proof of) Claim \ref{clm:robust}. First, we complete the proof of Theorem \ref{thm:subdel1} by analyzing the case of prompts which satisfy Definition \ref{def:consistent}:

\begin{claim*}
For constants $\epsilon_{sub}, \epsilon'_{sub}, \epsilon_{del}, \delta > 0$ such that $\epsilon_{sub} < \epsilon'_{sub} \delta / 16$, if $\PRC$ is robust against $\E_{\text{bin}}(p) \circ \E_{\text{del}}(1 - \epsilon_{del})$ for any $p \in [0, 1/2 - \epsilon_{sub})$, then, letting $f$ be any function family $f = \braces{f_n: \V(n) \rightarrow \Delta(\V(n))}_{n \in \Nat}$ and letting $\Pi$ be the set of prompts $\prompt$ that produce consistently distributed errors with respect to $\E_{sub}(f, \frac{1}{2} - \epsilon'_{sub})$, $\Watermark$ satisfies $(n \ell(n), \delta, 1)$-entropy-restricted substring robustness against the composition $\E_{sub}(f, \frac{1}{2} - \epsilon'_{sub}) \circ \E_{del}(1 - \epsilon_{del})$ over the set $\Pi$ of prompts.
\end{claim*}

\begin{proof}
This follows closely from Claims \ref{clm:expect} and \ref{clm:good}. Considering any output block defined in the proof of Claim \ref{clm:robust} as ``good" (i.e., being completely contained in the minimal-entropy substring and having at least a $\delta/4$ fraction of tokens with empirical entropy at least 1), it still holds after $\Watermark$ and $\E_{sub}(f, \frac{1}{2} - \epsilon'_{sub})$ that the expected fraction of tokens $t_i$ in that block for which $h_j(t_i) \neq c_i$ is no more than $1/2 - \epsilon'_{sub} \delta / 16$.

Given this, if the respective prompt satisfies Definition \ref{def:consistent}, then it must be the case that the probability $p$ for each token to have $h(t) \neq c_i$ can likewise be no greater than $1/2 - \epsilon'_{sub} \delta / 16$, or else the above statement would be trivially false. Thus, the modification made to the hash values $h(t)$ of the output tokens relative to the original PRC codeword $c$ by $\Watermark$ and $\E_{sub}(f, \frac{1}{2} - \epsilon'_{sub})$ must, according to Definition \ref{def:consistent}, be equivalent to the binary symmetric substitution channel $\E_{\text{bin}}(p)$ for some $p < \frac{1}{2} - \epsilon'_{sub} \delta / 16 < \frac{1}{2} - \epsilon_{sub}$.

Hence, when composing the above with the random deletion channel $\E_{del}(1 - \epsilon_{del})$, the transformation made to the codeword $c$ is equivalent to $\E_{\text{bin}}(p) \circ \E_{\text{del}}(1 - \epsilon_{del})$, which means that by the robustness property assumed of $\PRC$ the codeword must be recoverable with overwhelming probability.\footnote{This in fact applies to any block of the output as opposed to just ``good" blocks, since Definition \ref{def:consistent} likewise applies to the entire output.} 
\end{proof}

This, combined with Claim \ref{clm:sound}, suffices to prove Theorem \ref{thm:subdel1} and thus Corollary \ref{cor:subdel1}.

Finally, we complete the proof of Theorem \ref{thm:subdel2} by assuming a PRC with robustness against \emph{adversarial} substitutions and random deletions.

\begin{claim*}
For constants $\epsilon_{sub}, \epsilon'_{sub}, \epsilon_{del}, \delta > 0$ such that $\epsilon_{sub} < \epsilon'_{sub} \delta / 16$, if, for any $(1/2 - \epsilon_{sub})$-bounded binary substitution channel $\E$, $\PRC$ is robust against $\E \circ \E_{\text{del}}(1 - \epsilon_{del})$, then, for any function family $f = \braces{f_n: \V(n) \rightarrow \Delta(\V(n))}_{n \in \Nat}$, $\Watermark$ satisfies $(n \ell(n), \delta, 1)$-entropy-restricted substring robustness against the composition $\E_{sub}(f, \frac{1}{2} - \epsilon'_{sub}) \circ \E_{del}(1 - \epsilon_{del})$.
\end{claim*}

\begin{proof}
This follows nearly immediately from the proof of Claim \ref{clm:robust}. Therein we showed that, if $\epsilon_{sub} < \epsilon'_{sub} \delta / 16$, then with overwhelming probability there was some block in the output for which the modification made to the hash values $h(t)$ of the output tokens relative to the original PRC codeword $c$ by $\Watermark$ followed by $\E_{sub}(f, \frac{1}{2} - \epsilon'_{sub})$ was a $(1/2 - \epsilon_{sub})$-bounded binary substitution channel $\E$. 

In such a case, the transformation applied to $c$ by $\Watermark$ and $\E_{sub}(f, \frac{1}{2} - \epsilon'_{sub})$ subsequently followed by the deletion channel $\E_{del}(1 - \epsilon_{del})$ is equivalent to $\E \circ \E_{\text{del}}(1 - \epsilon_{del})$, against which $\PRC$ is robust by definition. Thus, with overwhelming probability, there exists some block of the output which is recoverable, as the construction of $\Detect$ guarantees that, irrespective of the number of deletions, the entirety of the respective block (that remains after deletion) will be decoded at some point in its execution.
\end{proof}

This, combined with Claim \ref{clm:sound}, suffices to prove Theorem \ref{thm:subdel2}.

\subsubsection{Edit Robustness}
\label{sec:edit}

Finally, we complete the proof of Theorem \ref{thm:edit}.

\begin{claim*}
For constants $\epsilon_{sub}, \epsilon_{edit}, \epsilon'_{edit}, \delta, \alpha > 0$ such that $\epsilon_{sub} < \delta / 64$ and $\epsilon'_{edit} < \epsilon_{edit} \delta \alpha / 17$, if, for any $(1/2 - \epsilon_{sub})$-bounded binary substitution channel $\E$ and any $\epsilon_{edit}$-bounded binary edit channel $\E'$, $\PRC$ is robust against $\E \circ \E'$, then $\Watermark$ satisfies $(\alpha m(n), \delta, 1)$-entropy-restricted substring robustness against any $\epsilon'_{edit}$-bounded edit channel.
\end{claim*}

\begin{proof}
To begin, we note that any $\epsilon'_{edit}$-bounded edit channel over the output of $\Watermark$ is trivially an $\epsilon'_{edit}$-bounded binary edit channel over the images under $h$ of the output tokens, as every insertion, substitution, or deletion to output tokens corresponds to at most one of the respective edit to the corresponding image (and possibly zero if, for instance, a token is replaced by a different token with the same hash).

Repeating the proof of Claim \ref{clm:robust}, if an at least $\alpha m(n)$-length substring of the output of $\Watermark$ has a $\delta$ fraction of tokens with sufficient empirical entropy, then, letting $k$ be the number of complete blocks contained within the substring (and noticing that $k \geq \alpha (m(n)/\ell(n)) - 1$, i.e., that $k$ is roughly equivalent to an $\alpha$ fraction of the number of blocks), at least $\delta k / 2$ blocks within that substring are ``good" blocks with at least a $\delta/4$ fraction of such tokens, which by Claim \ref{clm:expect} means that those blocks, with probability at least $1/4$ over the choice of functions $h$, have (with overwhelming probability) at most a $1/2 - \delta / 64$ error rate for embedding $c$.

Since the choice of $h$ is independent between blocks, these $1/4$ probabilities are independent. Furthermore, the expected number of ``good" blocks having at most a $1/2 - \delta / 64$ error rate is at least $\delta k / 8$; Hoeffding's inequality then gives us that the probability that the actual number of ``good" blocks meeting that criterion is less than $\delta k / 16$ (i.e., deviates from its expectation by at least $\delta k / 16$) is at most
\[\text{exp}(-2(\delta k / 16)^2/k) = \text{exp}(-2\delta^2 k / 256)\]
which is negligible in $k$ and thus (since the substring length is polynomial in $n$) in $n$. We refer to these blocks (``good" blocks having at most a $1/2 - \delta / 64$ error rate) henceforth as ``recoverable" blocks.

Thus, if $\PRC$ is robust to a $1/2 - \epsilon_\text{sub} > 1/2 - \delta / 64$ fraction of (arbitrary) substitutions in addition to an $\epsilon_{edit}$ fraction of (arbitrary) edits, the only way that the code $c$ cannot be successfully recovered from at least one of the ``recoverable" blocks is for \emph{all} of the ``recoverable" blocks to have at least an $\epsilon_{edit}$ fraction of edits made, which with overwhelming probability (and for sufficiently large $n$) is only possible if at least an $\epsilon_{edit} \delta \alpha / 17$ fraction of edits is made to the entire output (since at least a $\delta k / 16$ quantity of the blocks, or asymptotically at least a $\delta \alpha / 17$ fraction of the blocks, must have an $\epsilon_{edit}$ fraction of edits). Thus, $\Watermark$ is robust to any fraction of edits $\epsilon'_{edit} < \epsilon_{edit} \delta \alpha / 17$, as desired.

\end{proof}

This, combined with Claim \ref{clm:sound}, suffices to prove Theorem \ref{thm:edit}.

\paragraph{A note on parameterization.} We lastly note that, if instantiated using the PRC construction of \cite{cggm25} the specific value of $\epsilon_{edit}$ is quite small: using the parameterization given in Definition 5.8 of \cite{cggm25} requires that $\epsilon_\text{edit} = c \epsilon_\text{sub}^3 \text{log}(1/\epsilon_\text{sub})$
for a sufficiently small constant $c$, which effectively means that $\epsilon_\text{edit}$ is $\Tilde{O}(\delta^3)$ (since $\epsilon_{sub} < \delta / 64$) and, regarding the edit-robustness of $\Watermark$ itself, $\epsilon'_\text{edit}$ is $\Tilde{O}(\delta^4)$ (since $\epsilon'_{edit} < \epsilon_{edit} \alpha \delta / 17$). Obtaining better parameters for edit-robustness is theoretically possible but would require an analogous improvement in the underlying PRC construction.

\bibliographystyle{alpha}

\bibliography{cryptobib/abbrev3,cryptobib/crypto,watermarking}
\appendix

\end{document}